\documentclass[a4paper, 11pt]{article}
\usepackage[top=1.0 in,bottom=1.0 in,left=1.0 in,right=1.0 in]{geometry}
\usepackage{amsmath,amsthm}
\usepackage{amssymb}
\usepackage[colorlinks,linkcolor=blue,citecolor=blue]{hyperref}
\usepackage{float}
\usepackage{cite}
\usepackage{epic}
\usepackage{tikz}
\usepackage{graphicx}
\usepackage{accents}
\usepackage{slashed}
\newtheorem{definition}{Definition}[section]
\newtheorem{lemma}{Lemma}[section]
\newtheorem{theorem}{Theorem}[section]
\newtheorem{Proposition}{Proposition}[section]

\numberwithin{equation}{section}
\newtheorem{Remark}{Remark}[section]

\def \b#1{\overline{#1}}

\def \wt#1{\widetilde{#1}}
\def \wh#1{\widehat{#1}}
\def \wth#1{\widehat{\widetilde{#1}}}

\newcommand{\vep}{\varepsilon}
\newcommand{\p}{\prime}

\newcommand{\mf}{\mathbf}
\makeatletter

\newcommand{\Rmnum}[1]{\expandafter\@slowromancap\romannumeral #1@} \makeatother

\title{Elliptic soliton solutions:  $\tau$ functions, vertex operators and bilinear identities}

\author{Xing Li, ~~ Da-jun Zhang\footnote{Corresponding author. Email: djzhang@staff.shu.edu.cn}\\
	{\small\it Department of Mathematics,
		Shanghai University, Shanghai 200444,  P.R. China}}

\date{\today}

\begin{document}

\maketitle

\begin{abstract}
We establish a bilinear framework for elliptic soliton solutions
which are composed by the Lam\'e-type plane wave factors.
$\tau$ functions in Hirota's form are derived and vertex operators that generate such $\tau$ functions are presented.
Bilinear identities are constructed and an algorithm  to calculate residues and bilinear equations is formulated.
These are investigated in detail for the KdV equation and sketched for the KP hierarchy.
Degenerations by the periods of elliptic functions are investigated,
giving rise to the bilinear framework associated with
trigonometric/hyperbolic and rational functions.
Reductions by dispersion relation are considered by employing the so-called
elliptic $N$-th roots of the unity.
$\tau$ functions, vertex operators and bilinear equations of the KdV hierarchy and Boussinesq equation
are obtained from those of the KP.
We also formulate two ways to calculate bilinear derivatives involved with the Lam\'e-type plane wave factors,
which shows that  such type of plane wave factors result in quasi-gauge property of bilinear equations.

\vskip 6pt

\end{abstract}

\tableofcontents

%\newpage

\vskip 20pt

\section{Introduction}\label{sec-1}

The profound theory developed by Sato and his collaborators in 1980s
brings a deep insight on integrable systems \cite{MJD-book-1999}.
$\tau$ functions, vertex operators and bilinear identities together play a central role in this celebrated theory.
In particular, via vertex operators, $\tau$ functions and hence soliton equations are connected to affine Lie algebras.
These $\tau$ functions, generically,  are composed by   a plane wave factor (PWF) with a
\textit{linear exponential function}   
$e^{k t_1+k^2 t_2 + k^3 t_3 +\cdots}$.

In this paper we will develop Sato's theory for integrable systems,
aiming to establish a bilinear framework for the  $\tau$ functions, vertex operators and bilinear identities
that are associated with a Lam\'e-type PWF
\begin{equation}
\frac{\sigma(x+k)}{\sigma(x)\sigma(k)}e^{-\zeta(k)x+\zeta'(k)t_2-\frac{1}{2!}\zeta''(k)t_3+\cdots}.
\end{equation}
The Lam\'e function,
$y=\frac{\sigma(x+k)}{\sigma(x)\sigma(k)}e^{-\zeta(k)x}$,
is a doubly periodic function with respect to $k$,
bearing the name as it is a solution of the Lam\'e equation (the Schr\"odinger equation with an elliptic potential $\wp(x)$)
 \begin{equation}
 y''+(A+B\wp(x))y=0,
 \end{equation}
where $A=-\wp(k)$ and $B=-2$.
Here $\sigma, \zeta$ and $\wp$ are the Weierstrass $\sigma, \zeta$ and $\wp$ functions,
where $\wp(x)$ is an elliptic function, i.e. doubly periodic and meromorphic.
Elliptic curves can paly a role in integrable systems either as elliptic type solutions
or as elliptic deformations of the equations themselves, 
either way brings richer insight to integrable systems than trigonometric/hyperbolic and rational cases.
Apart from the famous finite-gap integration method developed by  Novikov, Matveev, Dubrovin, Its and Krichever
(see \cite{BBEIM-1994,Matveev} and the references therein),
a second pioneer work is \cite{AMM-CPAM-1977}
which extended the connection between the Korteweg-de Vries (KdV) equation and Calogero-Moser model
from rational  to elliptic case.
Soliton solutions based on the Lam\'e function have emerged in \cite{74Wahlqulst} in 1976 for the KdV equation.
In 2010 Nijhoff and Atkinson \cite{IMRN2010} developed a direct approach to obtain elliptic $N$-soliton solutions
for some quadrilateral equations that are consistent-around-cube and classified in \cite{ABS-2003}.
Their approach relies on Cauchy matrix and discrete (and elliptic) Lam\'e type PWFs.
The obtained solutions are termed as elliptic $N$-soliton solutions \cite{IMRN2010}.
Later, their approach was applied to
the lattice potential Kadomtsev-Petviashvili (KP) equation \cite{YN-JMP-2013}.
More recently, an elliptic direction linearisation approach was established in \cite{NSZ-2019},
and elliptic $N$-th roots of unity was introduced to construct
elliptic soliton solutions of the discrete Boussinesq type equations
and to deal with dimension reductions \cite{NSZ-2019}.

We shall now sketch the plan of this paper and describe main results in more detail.
The KdV equation will serve as our first introductory model  to bear details.
We will follow \cite{IMRN2010} and still use the term \textit{elliptic  soliton solutions},
although in continuous case these solutions are no longer elliptic (but still doubly periodic
with respect to parameters $k_j$ (see, e.g. Theorem \ref{T-3})
and expressed in terms of Weierstrass functions).
We will begin with an elliptic 1-soliton solution of the lattice potential KdV (lpKdV) equation.
By showing continuum limits of the equation and solution, we are able to have a full profile
as well as a comparison of the Lam\'e-type PWFs from fully discrete to continuous.
As a new feature, all these PWFs are no longer solutions to the linear part of the corresponding nonlinear equations.
This is different from the case of usual solitons composed by linear exponential functions.

Section \ref{sec-3} will play a role to present details that how a $\tau$ function for elliptic $N$-soliton solutions
is obtained from a Wronskian and from vertex operators,
how a bilinear identity is constructed and how explicit bilinear equations arise from the bilinear identity.
The KdV equation is still the model equation of this section.
We will begin by deriving its elliptic 1- and 2-soliton solutions from the bilinear KdV equation \eqref{bilinear-a}
using the standard Hirota's procedure, but the procedure is more complicated than the usual soliton case.
These two solutions are presented in Eq.\eqref{1ss-f1} and \eqref{2ss}.
Details of the derivation and some bilinear derivative formulae involved with the Lam\'e-type PWFs
are given in Appendix \ref{A-2}.
A key and new feature is the gauge property for bilinear derivatives of the usual soliton case
is not valid any longer for the Lam\'e-type PWFs,
and instead, we have quasi-gauge property (see Proposition \ref{P-B1}).
As a consequence, a KdV-type bilinear equation does not always admit an elliptic 2-soliton solution
and even elliptic 1-soliton.
This is also different from the usual soliton case where a KdV-type bilinear equation always has  a 2-soliton solution
\cite{Hir-1980,Hir-book-2004}.
The formula of $\tau$ function in Hirota's form for the elliptic $N$-soliton solution
is secured from a Wronskian that satisfies the bilinear KdV equation.
This formula is presented in Eq.\eqref{f-Hirota} in Theorem \ref{T-2}.
To obtain it, the quasi-gauge property  and some formulae and identities
of the Weierstrass functions are employed.
The vertex operator to generate such a $\tau$ function is given in Theorem \ref{T-3}.
After that, we will present a bilinear identity \eqref{bil-int-KdV} and its residue form \eqref{bil-res-KdV} in Theorem \ref{T-4}.
The identity is constructed by using double-periodicity of the integrand
and implementing the integration around the fundamental period parallelogram.
It turns out that the integrand has $2N$ simple poles and an essential singularity at $q=0$ (mod period lattice).
Similar to the usual soliton case, the integral bilinear identity equals to the residue of the integrand at $q=0$,
but the way to achieve the residue is not straightforward at all.
We will develop an algorithm for this matter in Sec.\ref{sec-3-4}
and a practical formula for calculating residues as well as bilinear equations 
is presented Eq.\eqref{bil-D-KdV-4} in Theorem \ref{T-5}.
After the exploration of the KdV equation with necessary details,
we will move to the KP equation in  Sec.\ref{sec-4}
and sketch the main results  in Theorems \ref{T-7}, \ref{T-8}, \ref{T-9} and \ref{T-10}.

In section \ref{sec-5} we will discuss period degenerations of the elliptic soliton solutions when the
discriminant $\Delta=g_2^3-27g_3^2=0$.
This will give rise to soliton solutions of trigonometric/hyperbolic type and rational type.
The degenerations are straightforward.
That is to say, one can directly substitute the degenerated Weierstrass functions (see Proposition \ref{P-A2})
into the $\tau$ functions and bilinear equations we obtain in sections \ref{sec-3} and \ref{sec-4}.
The degenerated results for the KP hierarchy are given in Theorem \ref{T-11} and \ref{T-12}.
Three types of PWFs of the KP hierarchy are given in \eqref{PWF-KP-e},  \eqref{PWF-KP-t} and \eqref{PWF-KP-r},
respectively.
Note that Theorem \ref{T-11} presents a more concise expression for the trigonometric/hyperbolic-type $\tau$ function and
the associated vertex operator,
which allows a direct replacement of $\sigma(x)$ and $\zeta(x)$ by $\sin(\alpha x)$ and $\alpha \cot(\alpha x)$ respectively.
In section \ref{sec-5} we will also investigate reductions by dispersion relations
(corresponding to periodic reductions of the usual soliton case).
Elliptic $N$-th roots of the unity (see \cite{NSZ-2019} and Definition \ref{D-1} in this paper) will be used.
However, different from the usual soliton case,
when $N\geq 3$, the  elliptic $N$-th roots of the unity are not simultaneously the elliptic $(kN)$-th roots of the unity
where $k\in \mathbb{N}$, (see Remark \ref{R-A1} in Appendix \ref{A-1}).
This means one cannot get elliptic soliton solutions for the Gel'fand-Dickey (with $N\geq 3$) hierarchy
from those of the KP hierarchy by reduction using elliptic $N$-th roots of the unity.

We have introduced the plan of our paper as well as the main results and some new features associated with
the Lam\'e-type PWFs.
The paper also contains a section where we will present conclusions and mention some further topics
based on the framework of this paper.
In addition, there are three appendices,
which include a collection of the Weierstrass functions and the related properties and identities,
some calculating formulae involved with Hirota's bilinear operator and the Lam\'e-type PWFs,
and proofs for the elliptic $N$-soliton solutions in Wronskian forms that satisfy respectively the bilinear KdV equation
and KP equation.

\section{Lam\'e-type plane wave factors}\label{sec-2}

PWF is an elementary block of $N$-soliton solutions.
In this section we begin by exploring PWFs and dispersion relations of elliptic solitons, 
for fully discrete, semi-discrete and continuous cases.
We will consider usual 1-soliton solution and elliptic 1-soliton solution of the lpKdV equation
and implement continuum limits of both the equation and solution,
so that one can make a comparison for the usual and elliptic cases.

Recalling the KdV equation (with scaled coefficients for our convenience)
\begin{equation}\label{kdv}
u_t=\frac{3}{2}uu_x+\frac{1}{4}u_{xxx}
\end{equation}
and its potential form $(u=v_x)$
\begin{equation}\label{pkdv}
v_t=\frac{3}{4}v_x^2+\frac{1}{4}v_{xxx},
\end{equation}
which admits 1-soliton solution
\begin{equation}
v=\frac{4k e^{2k x+2k^3 t}}{1+e^{2k x+2k^3 t}}.
\end{equation}
The PWF is
\begin{equation}\label{PWF-kdv}
\rho(k)=e^{2k x+2k^3 t},
\end{equation}
which is a solution of the linear part of the (potential) KdV equation and indicates the dispersion relation of the equation.

The lpKdV equation reads \cite{bookHJN,NC-AAM-1995,NCW-LNP-1985}
\begin{equation}\label{lpkdv}
(w-\wh{\wt{w}})(\wh{w}-\wt{w})=p^2-q^2,
\end{equation}
where we adopt notations
\[w=w(n,m),~~ \wt w=w(n+1,m),~~ \wh w=w(n,m+1),~~ \wth w=w(n+1,m+1),\]
$n,m\in \mathbb{Z}$, $p$ and $q$ are spacing parameters of the $n$-direction and $m$-direction, respectively.
This equation has a background solution $w_0=pn+qm+c$ and a usual 1-soliton solution \cite{09HZ}
\begin{equation}\label{1ss}
w=w_0+\frac{k (1-\rho)}{1+\rho},
\end{equation}
where
\begin{equation}\label{PWF-lpkdv}
\rho=\biggl(\frac{p+k}{p-k}\biggr)^n\biggl(\frac{q+k}{q-k}\biggr)^m\rho_{00}
\end{equation}
is the PWF.
Here $c, k$ and $\rho_{00}$ are constants.
Removing the background $w_0$ by introducing $v=w-w_0$,
the lpKdV equation \eqref{lpkdv} is converted to
\begin{equation}\label{lpkdv-v}
(v-\wh{\wt{v}}-p-q)(\wh{v}-\wt{v}-p+q)=p^2-q^2.
\end{equation}
The PWF \eqref{PWF-lpkdv} solves the linear part of the above equation.

With new  parametrizations
\begin{equation}\label{pq-para}
p^2=\wp(\delta)-e_1,~~ q^2=\wp(\vep)-e_1,
\end{equation}
 the lpKdV equation \eqref{lpkdv} allows a background solution
 \begin{equation}
w_0 =\zeta(\xi)-n\zeta(\delta)-m\zeta(\vep)-c_0
\end{equation}
where
\begin{equation}
\xi=n\delta + m \vep,
\end{equation}
$e_1,   c_0 \in \mathbb{C}$,
and $\delta, \varepsilon$ serve as lattice parameters.
For the Weierstrass functions $\sigma(x), \zeta(x)$ and $\wp(x)$ 
and related notations and properties please refer to Appendix \ref{A-1}.
The elliptic 1-soliton solution of the lpKdV equation is  \cite{IMRN2010}
\begin{equation}\label{1ss-ell}
	w=w_0+\frac{\eta_{-k}(\xi)+\eta_{k}(\xi) \rho}{1+\rho},
\end{equation}
where
\begin{equation}\label{eta}
\eta_x(y)=\zeta(x+y)-\zeta(x)-\zeta(y),
\end{equation}
and the PWF is
\begin{equation}\label{PWF-ell}
\rho=\frac{\sigma(k+\xi)}{\sigma(k-\xi)}\left(\frac{\sigma(k-\delta)}{\sigma(k+\delta)}\right)^n
\left(\frac{\sigma(k-\vep)}{\sigma(k+\vep)}\right)^m \rho_{00},
\end{equation}
with $k, \rho_{00}\in \mathbb{C}$.
Again, removing the background $w_0$ from \eqref{lpkdv} by $v=w-w_0$ yields
\begin{equation}\label{lpkdv-v-ell}
 (v-\wh{\wt{v}}+\chi_{\delta,\vep}(\xi))(\wh{v}-\wt{v}-\chi_{-\delta,\vep}(\xi+\delta))
=\wp(\delta)-\wp(\vep),
\end{equation}
where
\begin{equation}\label{chi}
\chi_{\delta,\vep}(\gamma)=\zeta(\delta)+\zeta(\vep)+\zeta(\gamma)-\zeta(\delta+\vep+\gamma).
\end{equation}
Equation \eqref{lpkdv-v-ell} admits a solution
\begin{equation}\label{v-1ss}
	v=\frac{\eta_{-k}(\xi)+\eta_{k}(\xi) \rho}{1+\rho}
\end{equation}
with PWF \eqref{PWF-ell}.
Note that for given $n,m$ and constant $\rho_{00}$ that are independent of $(k, \delta, \vep)$,
the PWF $\rho$ and $\eta_{\pm k}(\xi)$ are elliptic functions of $(k, \delta, \vep)$, and so is $v$ given above.
However, the PWF \eqref{PWF-ell} is no longer a solution of
the linear part of the equation \eqref{lpkdv-v-ell}.

To show the Lam\'e-type PWFs in semi-discrete and continuous form,
we consider continuum limits of the lpKdV equation \eqref{lpkdv-v-ell}
together with its elliptic soliton solution \eqref{v-1ss}.
Let $m \rightarrow\infty,~\vep \rightarrow 0$ while $\mu=m \vep$ be finite.
Noticing those Laurent series listed in \eqref{pzs-expan} and
\begin{align*}
&\chi_{\delta,\vep}(\xi)=\frac{1}{\vep}-\eta_{\delta}(\mu+n\delta)+\vep\wp(\mu+(n+1)\delta)
+\frac{\vep^2}{2}\wp'(\mu+(n+1)\delta)+O(\vep^3),\\
&\chi_{-\delta,\vep}(\xi+\delta)=\frac{1}{\vep}+\eta_{\delta}(\mu+n\delta)+\vep\wp(\mu+n\delta)
+\frac{\vep^2}{2}\wp'(\mu+n\delta)+O(\vep^3),
\end{align*}
in continuum limits the lpKdV equation \eqref{lpkdv-v-ell} yields
the semi-discrete pKdV equation (with a $n$-dependent coefficient $\eta_{\delta}(\mu+n\delta)$)
\begin{equation}\label{pkdv-sd}
\partial_{\mu}(v+\wt{v})+(\wt{v}-v)^2+2\eta_{\delta}(\mu+n\delta)(\wt{v}-v)=0,
\end{equation}
which admits an elliptic 1-soliton solution
\begin{equation}\label{v-1ss-sd}
	v=\frac{\eta_{-k}(n\delta+\mu)+\eta_{ k}(n\delta+\mu) \rho}{1+\rho}
\end{equation}
where the  PWF is (with $\rho_{0}\in \mathbb{C}$)
\begin{equation}\label{PWF-ell-sd}
\rho=\frac{\sigma(k+n\delta+\mu)}{\sigma(k-n\delta-\mu)}\left(\frac{\sigma(k-\delta)}{\sigma(k+\delta)}\right)^n
 e^{-2\zeta(k)\mu}\rho_{0}.
\end{equation}
Strictly speaking, this PWF is doubly periodic with respect to $k$ but not elliptic
as there is an essential singularity at $k=0$ due to $e^{-2\zeta(k)\mu}$.
However, we would like to inherit the term \textit{elliptic $N$-soliton solutions} introduced in \cite{IMRN2010}.
Note also that the PWF does not solve the linear part of Eq.\eqref{pkdv-sd} either.
In the full continuum limit, first, we let $n \rightarrow\infty,~ \delta \rightarrow 0$ while
$\nu=n\delta$ be finite, and then introduce
$x=\mu+\nu, ~ t=\frac{1}{3}\delta^2 \nu$.
The resulting equation with coordinates $(x,t)$ is
\begin{equation}\label{pkdv-ell}
	v_t-\frac{3}{2}v_x^2+3\wp(x)v_x-\frac{1}{4}v_{xxx}=0,
\end{equation}
and its elliptic 1-soliton solution takes a form
\begin{equation}\label{v-1ss-cc}
v= \frac{\eta_{-k}(x)+\eta_{ k}(x) \rho}{1+\rho},
\end{equation}
where the PWF for the continuous elliptic soliton solution is
\begin{equation}\label{PWF-ell-cc}
\rho=\frac{\sigma(k+x)}{\sigma(k-x)}
 e^{-2\zeta(k)x+\wp'(k)t+\xi^{(0)}},
\end{equation}
with parameter $\xi^{(0)}\in \mathbb{C}$  independent of $k$ or being a doubly periodic function of $k$.
Note that employing the transformation
\begin{equation}
\bar{v}=2(v+\zeta(x)+\frac{1}{8}g_2 t)
\end{equation}
one can convert Eq.\eqref{pkdv-ell} into the usual pKdV equation (i.e. \eqref{pkdv})
\begin{equation}\label{pkdv-ell-1}
\bar{v}_t-\frac{3}{4}\bar{v}_x^2-\frac{1}{4}\bar{v}_{xxx}=0.
\end{equation}
Besides, the nonpotential form of Eq.\eqref{pkdv-ell} is $(u=v_x)$
\begin{equation}\label{kdv-ell}
	u_t-3 u u_x+3\wp(x)u_x+3\wp'(x)u -\frac{1}{4}u_{xxx}=0,
\end{equation}
which, by transformation  $u \rightarrow \frac{1}{2}u+\wp(x)$,
is written as the usual KdV equation \eqref{kdv}.
However, the PWF \eqref{PWF-ell-cc} is not a solution of the
linear part of any of equations, \eqref{pkdv-ell} or \eqref{pkdv-ell-1} or \eqref{kdv-ell} or \eqref{kdv}.
Note that the elliptic 1-soliton solution $u=v_x$ with \eqref{v-1ss-cc} emerged in \cite{74Wahlqulst}.

Now let us make a comparison for the two PWFs, \eqref{PWF-ell-cc} and \eqref{PWF-kdv},
i.e. the PWFs for elliptic solitons and usual solitons.
Considering the exponential parts of them, asymptotically,
it follows from \eqref{pzs-expan} that
\[e^{-2\zeta(k)x+\wp'(k)t+\xi^{(0)}} \sim ~ e^{-\frac{2}{k}x-\frac{2}{k^3} t},\]
which  corresponds to the dispersion relation in \eqref{PWF-kdv}.
This observation motivates us to introduce a general Lam\'e-type PWF (the extended Lam\'e function)
\begin{equation}
\rho= \Phi_x(k) \exp\Bigl(-\zeta(k) t_1 + \zeta'(k) t_2+\cdots + \frac{(-1)^j}{(j-1)!}\zeta^{(j-1)}(k) t_j+\cdots\Bigr ),
\end{equation}
which is an elliptic analogue of the usual one
\begin{equation}
\rho= \exp\bigl(k t_1 + k^2 t_2 +\cdots + k^j t_j+ \cdots\bigr ),
\end{equation}
where $t_1=x$ and
\begin{equation}\label{Phi}
 \Phi_x(k)=\frac{\sigma(k+x)}{\sigma(x)\sigma(k)}.
\end{equation}
Note that the doubly-periodic feature of the PWF \eqref{PWF-ell-cc} can also be illustrated in its alternative form
\begin{align}\label{rho-ell}
\rho   &=\exp\left(  \xi^{(0)}+2\wp'(k)t-\sum_{n=1}^{\infty}\frac{2}{(2n+1)!}\wp^{(2n-1)}(k)x^{2n+1}\right ).
\end{align}

For the KdV equation \eqref{kdv}, its elliptic 1-soliton solution can be written as (cf. Eq.\ref{1ss-f1})
\begin{equation}
u=-2\wp(x)+2(\ln (1+\Phi_x(2k) e^{-2\zeta(k)x+\wp'(k)t+\xi^{(0)}}))_{xx},
\label{ell-1ss-kdv}
\end{equation}
where the $-2\wp(x)$ is a 1-gap and 1-genus solution in light of the so-called Dubrovin's equations
in finite-gap integration  \cite{Dub-FAA-1975,DN-JETP-1974} (also see \cite{Ince-1940} by Ince),
but the whole solution \eqref{ell-1ss-kdv} is a doubly periodic function of $k$
(not periodic  with respect to $x$).

Noting that $(\wp(k), \wp'(k))$, $k\in \mathbb{D}$ (see Fig.\ref{Fig-1}) are points on the elliptic curve \eqref{ell-cur},
along the line of \cite{AMM-CPAM-1977},
we can say that the elliptic soliton solution corresponds to the torus \eqref{ell-cur},
while its degenerations by fixing $g_2=\frac{4}{3}(\frac{\pi}{2 w_1})^4, g_3=\frac{8}{27}(\frac{\pi}{2 w_1})^6$
and $g_2=g_3=0$ (i.e. degenerations by periods)
correspond to a cylinder and Riemann sphere, respectively, cf. \cite{AMM-CPAM-1977}.

\section{$\tau$ function, vertex operator and bilinear identity: KdV}\label{sec-3}

We will extend the obtained elliptic 1-soliton solution of the KdV equation to its  elliptic $N$-soliton solution
and then establish a bilinear framework for such type of solutions.
The framework will consist of $\tau$ function in Hirota's form, a vertex operator for generating the $\tau$ function,
a bilinear identity and an algorithm for calculating residues  that gives rise to bilinear equations.

\subsection{$\tau$ function of elliptic $N$-soliton solutions}\label{sec-3-1}

\subsubsection{Bilinear form and elliptic 1- and  2-soliton solutions}\label{sec-3-1-1}

We begin by exploring Hirota's procedure to calculate elliptic 1- and 2-soliton solutions
for a bilinear KdV equation.
The potential KdV equation \eqref{pkdv-ell-1} can be converted into a bilinear form
\begin{equation}\label{bilinear-a}
(D_x^4-4D_xD_t-12\wp(x)D_x^2)\tau\cdot\tau=0
\end{equation}
via the transformation
\begin{equation}\label{eq:v}
\bar v =2\zeta(x)+\frac{1}{4}g_2t+2(\ln \tau)_x,
\end{equation}
where $D$ is Hirota's bilinear operator defined by\cite{74Hirota}
\[D_t^mD_x^n f\cdot g=(\partial_t-\partial_{t'})^m(\partial_x-\partial_{x'})^nf(t,x)g(t',x')|_{t'=t,x'=x},~~ m,n=0,1,2\cdots.\]
Equation \eqref{bilinear-a} is also a bilinear form of the KdV equation \eqref{kdv} while the transformation is
\begin{equation}\label{eq:u}
u= -2\wp(x)+2(\ln \tau)_{xx}.
\end{equation}
Both \eqref{eq:v} and \eqref{eq:u} have nonzero backgrounds.
An alternative bilinear form for the KdV equation is
\begin{equation}\label{bilinear-aa}
	(D_x^4-4D_xD_t-g_2)\tau' \cdot \tau'=0,
\end{equation}
while the associated transformations are
\begin{equation}
	u= 2(\ln \tau')_{xx}, \quad \bar v =\frac{1}{4}g_2t+2(\ln \tau')_x.
\end{equation}

By direct calculation (see Appendix \ref{A-2}), one can find that  Eq.\eqref{bilinear-a} admits the following solutions,
\begin{equation}\label{1ss-f1}
\tau=	f_1=1+\rho_1(x,t) 
=1+\Phi_x(2k_1)e^{\xi_1},
\end{equation}
	and
\begin{align}
\tau=	f_2 &=1+\rho_1(x,t)+\rho_2(x,t)+f^{(2)}(x,t) \nonumber \\
	&=1+\Phi_x(2k_1) e^{\xi_1}
         +\Phi_x(2k_2) e^{\xi_2}+A_{12} \frac{\sigma(x+2k_1+2k_2)}{\sigma(x)\sigma(2k_1)\sigma(2k_2)}e^{\xi_1+\xi_2},
                \label{2ss}
\end{align}
where
	\begin{equation}\label{PWF-2ss}
    \rho_i(x,t)=\Phi_x(2k_i) e^{\xi_i},~~
	\xi_i=-2\zeta(k_i)x+\wp'(k_i)t+\xi^{(0)}_i,
	~~ A_{12}=\frac{\sigma^2(k_1-k_2)}{\sigma^2(k_1+k_2)},
	\end{equation}
$k_i, \xi^{(0)}_i \in \mathbb{C}$.
These are formally similar to the usual 1-soliton and 2-soliton solutions of the KdV equation
but there is an essential difference in 2-soliton case:
the last term $f^{(2)}$ in $f_2$ is $A_{12}e^{4\zeta(k_1)k_2}\rho_1(x+2k_2,t)\rho_2(x,t)$,
rather than $A_{12}\rho_1(x,t)\rho_2(x,t)$ as in a usual two-soliton solution.
In Appendix \ref{A-2} we provide details of deriving $f_1$ and $f_2$,
as well as some formulae for higher order bilinear derivatives
and properties (e.g. the quasi-gauge property, see Proposition \ref{P-B1})
involved with the Lam\'e-type PWF $\rho_i$.
We also remark that
it is well known  a KdV-type bilinear equation (with constant coefficients)
always admits 1-soliton solution and 2-soliton solution \cite{Hir-1980},
however, such a convention does not hold even for admitting elliptic 1-soliton solution.

A $\tau$ function in Hirota's form for elliptic $N$-soliton solution is needed to introduce vertex operator.
However, for higher order elliptic soliton solutions, the calculation is much more complicated.
Next, we will first present a $N$-soliton solution in terms of Wronskian,
from which we can secure the $\tau$ function in Hirota's form.

\subsubsection{$\tau$ function in Wronskian form}\label{sec-3-1-2}

Introduce a $N$-th order column vector
\begin{equation}\label{vphi}
\varphi=(\varphi_1,\varphi_2,\cdots,\varphi_N)^T,
\end{equation}
where $\varphi_j=\varphi_j(x,t)$ are functions of $(x,t)$.
A $N$-th order Wronskian is defined as
\[f= |\varphi, \partial_x \varphi, \partial_x^2 \varphi,    \cdots,  \partial_x^{N-1}\varphi|
=|0,1,2,\cdots,N-1|=|\widehat{N-1}|,
\]
where we  employ the conventional shorthand introduced in \cite{FreN-1983}.
For an elliptic $N$-soliton solution of the KdV equation, we have the following.

\begin{theorem}\label{T-1}
The bilinear equation \eqref{bilinear-a} admits a Wronskian solution
\begin{equation}\label{tau-W}
\tau=|\widehat{N-1}|
\end{equation}
composed by vector $\varphi=(\varphi_1,\varphi_2,\cdots,\varphi_N)^T$ where each element  $\varphi_j$ satisfies
\begin{subequations}\label{Wrons-cond}
\begin{align}
&\varphi_{j,xx}=(\wp(k_j)+2\wp(x))\varphi_j, \label{phi-jx}\\
& \varphi_{j,t}=\mf{\varphi}_{j,xxx}-3\wp(x)\varphi_{j,x}-\frac{3}{2}\wp'(x)\varphi_j, \label{phi-jt}
\end{align}
\end{subequations}
for $j=1,2,\cdots,N$ and $k_j\in \mathbb{C}$.
A general solution to the above equations is
\begin{subequations}\label{phi-j}
\begin{equation}\label{phi-j-Wr}
 \varphi_{j}=a_j^+\varphi_j^++ a_j^-\varphi_j^-,
\end{equation}
where $\varphi_j^{\pm}$ are Lam\'{e}   functions
\begin{equation}\label{phi-j-gamma}
\varphi_j^{\pm}=\Phi_{x}(\pm k_j) e^{\mp \gamma_j},~~ \gamma_j=\zeta(k_j)x-\frac{1}{2}\wp'(k_j)t+\gamma_j^{(0)},
\end{equation}
\end{subequations}
where $a_j^{\pm}, k_j, \gamma_j^{(0)} \in \mathbb{C}$,  $\Phi_x(k)$ is defined in \eqref{Phi},
and in practice, $k_j$ takes value in the fundamental period parallelogram $\mathbb{D}$ of the
Weierstrass $\wp$ function (see Fig.\ref{Fig-1}).
\end{theorem}

The proof will be sketched in Appendix \ref{A-3}.
Note that such a solution in Wronskian form for the KdV equation can be alternatively obtained
using the Darboux transformation by taking $u=-2\wp(x)$ as a seed solution 
and assigning a proper dispersion relation (see \cite{MatS-1991}),
but we do need to have a $\tau$ function that serves for elliptic $N$-soliton solutions
and  satisfies a definite bilinear KdV equation.\footnote{Due to the quasi-gauge property (see Proposition \ref{P-B1})
of bilinear derivatives with respect to the Lam\'e-type PFWs,
it is necessary have some $\tau$ function to satisfy a definite bilinear equation.}

\subsubsection{$\tau$ function in Hirota's form}\label{sec-3-1-3}

To convert Wronskian \eqref{tau-W} into Hirota's form, we first investigate the Wronskian composed by
$\varphi^-=(\varphi_1^-,\varphi_2^-,\cdots, \varphi_N^-)^T$ and its derivatives,
where $\{\varphi_j^-\}$ are defined as in \eqref{phi-j-gamma}.
Such a Wronskian can be written as an explicit form.

\begin{lemma}\label{L-1}
For the forementioned $\varphi^-$, we have
\begin{align}
   &|\varphi^-,~ \partial_x \varphi^- ,~\partial_x^2 \varphi^-,~ \cdots  ,~\partial_x^{N-1}\varphi^- |\nonumber \\
=\,&(-1)^N\frac{\sigma(x-\sum_{i=1}^Nk_i)}{\sigma(x)}\cdot
\frac{\prod_{1\leq i<j\leq N}\sigma(k_i-k_j)}{\sigma^{N}(k_1)\cdots\sigma^{N}(k_N)}
\exp\left( \sum_{i=1}^N \gamma_j\right).\label{W-phi-}
\end{align}
\end{lemma}

\begin{proof}
For convenience we introduce notations $\mathbf{k}=(k_1, k_2, \cdots, k_N)^T$,
$f(\mathbf{k})=(f(k_1), f(k_2),$ $ \cdots, f(k_N))^T$,
$f(\mathbf{k}) g(\mathbf{k})=(f(k_1)g(k_1), f(k_2)g(k_2), \cdots, f(k_N)g(k_N))^T$,
and we consider the Wronskian
\begin{equation}\label{f-}
f^-=|\Phi_{x}(-\mathbf{k}) e^{\zeta(\mathbf{k})x},~ \partial_x(\Phi_{x}(-\mathbf{k}) e^{\zeta(\mathbf{k})x} ) ,~
\cdots ,~\partial_x^{N-1}(\Phi_{x}(-\mathbf{k}) e^{\zeta(\mathbf{k})x} ) |,
\end{equation}
where for conciseness we have dropped off $\wp'(k_j)t$ and $\gamma_j^{(0)}$ in $\gamma_j$
since the structure of the Wronskian is irrelevant to time.
For each $\varphi_j^-$ we have
\[ \partial_x \varphi_j^- =\eta_x(-k) \varphi_j^- ,\]
where  $\eta_x(k)$ is defined as \eqref{eta}.
In addition,  $\varphi_j^-$ is a Lam\'e   function, satisfying \eqref{phi-jx},
which indicates that
\[ \partial_x^{n}\phi^-_j=(\wp(k_j)+2\wp(x))\,\partial_x^{n-2}\phi^-_j
+2\sum_{i=1}^{n-2}\left(\begin{array}{cc}
                                      n-2\\  i
                                      \end{array}\right)
                                      (\partial_x^{i}\wp(x))\partial_x^{n-2-i}\phi^-, ~~(n\geq 2).	
\]
Using the above relations we can replace the column 
$\partial_x^j(\Phi_{x}(-\mathbf{k}) e^{\zeta(\mathbf{k})x})$ in \eqref{f-},
and after simplification it turns out that
\begin{align}
f^-=~&\left(\exp\sum_{i=1}^N\zeta(k_i)x \right)\left(\prod_{j=1}^N\Phi_x(-k_j)\right)  \nonumber\\
&\times |1 ,  ~\eta_{x}(-\mathbf{k}), ~ \wp(\mathbf{k}), ~ \wp(\mathbf{k})\eta_{x}(-\mathbf{k}),~
  \wp^2(\mathbf{k}),~  \wp^2(\mathbf{k})\eta_x(-\mathbf{k}),~\cdots ,~
  \wp^{\left[\frac{N-1}{2}\right]}(\mathbf{k})h_1(\mathbf{k},x) |, \label{f-1}
\end{align}
where in the last column $h_1(\mathbf{k},x)$ stands for
\[h_1(\mathbf{k},x)=\left\{\begin{array}{ll} 1,& N ~\mathrm{odd},\\ \eta_{x}(-\mathbf{k}), & N ~ \mathrm{even},
                               \end{array}\right.
\]
and $[x]$ is the floor function of $x$.

Next,  for the column $\wp^n(\mathbf{k})\eta_x(-\mathbf{k})$ in \eqref{f-1}, in light of the relation \eqref{eq:add-2},
we have (for $n\geq 1$)
\[\wp^n(\mathbf{k})\eta_x(-\mathbf{k})=-\frac{1}{2}\wp^{n-1}(\mathbf{k}) \wp'(\mathbf{k})
-\frac{1}{2}\wp^{n-1}(\mathbf{k}) \wp'(x)+\wp^{n-1}(\mathbf{k})\eta_x(-\mathbf{k})\wp(x),
\]
where the last two terms on the right hand side will be eliminated by those front columns in \eqref{f-1}.
We can examine all such columns in \eqref{f-1} successively from right to left.
As a result, we are able to have $f^-$ in the form
\begin{align*}
f^-=&\left(-\frac{1}{2}\right)^{\left[\frac{N}{2}\right]-1}
\left(\exp\sum_{i=1}^N\zeta(k_i)x \right)\left(\prod_{j=1}^N\Phi_x(-k_j)\right)\\
&\times |1 ,  \eta_{x}(-\mathbf{k}),  \wp(\mathbf{k}),  \wp'(\mathbf{k}),
  \wp^2(\mathbf{k}),  \wp(\mathbf{k})\wp'(\mathbf{k}),
  \wp^3(\mathbf{k}),  \wp^2(\mathbf{k})\wp'(\mathbf{k}),~\cdots ,
  \wp^{\left[\frac{N-1}{2}\right]-1}(\mathbf{k})h_2(\mathbf{k}) |, \label{f-2}
\end{align*}
where in the last column $h_2(\mathbf{k})$ is
\[h_2(\mathbf{k})=\left\{\begin{array}{ll} \wp(\mathbf{k}),& N ~\mathrm{odd},\\
                                                               \wp'(\mathbf{k}), & N ~ \mathrm{even}.
                               \end{array}\right.
\]
By virtue of the fact that $(\wp(x), \wp'(x))$ is a point on the elliptic curve \eqref{ell-cur}, i.e.
\[(\wp'(k))^2=4 \wp^3(k)-g_2\wp(k)-g_3,\]
we know that both $\wp^{(2n-2)}(x)$ and $\frac{\wp^{(2n+1)}(x)}{\wp'(x)}$
can be expressed as a linear combination of $\{\wp^s(x)\}$ with $s=n, n-2, n-3, \cdots, 2,1,0$.
Then,
we are led to
\begin{align*}
f^-=~&
\left(\exp\sum_{i=1}^N\zeta(k_i)x \right)\left(\prod_{j=1}^N\Phi_x(-k_j)\right)
\frac{1}{1! 2! \cdots (N-2)!}\\
%\nonumber\\
&\times |1 ,  ~\eta_{x}(-\mathbf{k}), ~ \wp(-\mathbf{k}), ~ \wp'(-\mathbf{k}),~
  \wp''(-\mathbf{k}),~  \wp'''(-\mathbf{k}),~ \cdots ,~
  \wp^{(N-3)}(-\mathbf{k}) |, %\label{f-3}
\end{align*}
which is further written into
\begin{align}
f^-=~&
\left(\exp\sum_{i=1}^N\zeta(k_i)x \right) \frac{(-1)^{N-1}\Phi_x(-k_1-\cdots - k_N)}{1! 2! \cdots (N-2)!(N-1)!}
\nonumber\\
&\times |1 ,   ~ \wp(-\mathbf{k}), ~ \wp'(-\mathbf{k}),~
  \wp''(-\mathbf{k}),~  \wp'''(-\mathbf{k}),~ \cdots ,~
  \wp^{(N-2)}(-\mathbf{k}) |, 
\end{align}
where use has been made of relation \eqref{eq:ND16}.
Then, employing the elliptic van der Monde determinant formula \eqref{eq:FS-1}, we have
\begin{equation}
f^-=(-1)^N\frac{\sigma(x-\sum_{i=1}^Nk_i)}{\sigma(x)}
\frac{\prod_{1\leq i<j\leq N}\sigma(k_i-k_j)}{\sigma^{N}(k_1)\cdots\sigma^{N}(k_N)}
\exp\left(\sum_{i=1}^N\zeta(k_i)x\right),
\end{equation}
which yields \eqref{W-phi-}.

\end{proof}

Next, in order to obtain the $\tau$-function in Hirota's form, we consider the Wronskian \eqref{tau-W}
composed specially by an elementary column vector (cf.\eqref{phi-j-Wr})
\begin{equation}\label{phi-j-hirota}
 \varphi_{j}= \varphi_j^+ + (-1)^j \varphi_j^-,
\end{equation}
where $\varphi_j^{\pm}$ are defined by \eqref{phi-j-gamma}.
The corresponding Wronskian
\begin{equation}\label{tau-W2}
\tau=|\widehat{N-1}|
\end{equation}
can be split and then written as a sum of $2^N$ distinct
Wronskians, each of which is generated by the elementary column vector of the following form,
\begin{equation}\label{phi-eps}
\varphi=(\phi_1, \phi_2, \cdots, \phi_N)^T,~~
\phi_j =(\varepsilon_j)^j\Phi_x(\varepsilon_j k_j)e^{-\varepsilon_j \gamma_j}
\end{equation}
where $\{\varepsilon_1, \varepsilon_2, \cdots, \varepsilon_N\}$ run over $\{1,-1\}$.
In light of Lemma \ref{L-1}, the Wronskian generated by the above $\varphi$ is
\begin{align}
\tau_{\boldsymbol{\varepsilon}}
&= (-1)^\frac{N(N-1)}{2}\prod^{N}_{j=1}(\varepsilon_j)^j\cdot
\frac{\sigma(x+\sum_{i=1}^N \varepsilon_i k_i)}{\sigma(x)}\cdot
\frac{\prod_{1\leq i<j\leq N}\sigma(\varepsilon_i k_i-\varepsilon_j k_j)}
{\sigma^{N}(\varepsilon_1 k_1)\cdots\sigma^{N}(\varepsilon_N k_N)}
\exp \left(-\sum_{i=1}^N\varepsilon_i\gamma_i\right)\nonumber\\
&= (-1)^\frac{N(N-1)}{2}\cdot
\frac{\sigma(x+\sum_{i=1}^N \varepsilon_i k_i)}{\sigma(x)}\cdot
\frac{\prod_{1\leq i<j\leq N}\varepsilon_i\sigma(\varepsilon_i k_i-\varepsilon_j k_j)}
{\sigma^{N}(k_1)\cdots\sigma^{N}(k_N)}
\exp \left(-\sum_{i=1}^N\varepsilon_i\gamma_i\right),\label{W-phi-eps}
\end{align}
where $\boldsymbol{\varepsilon}$ indicates cluster
$\boldsymbol{\varepsilon}=\{\varepsilon_1, \varepsilon_2, \cdots, \varepsilon_N\}$.
Introduce length of $\boldsymbol{\varepsilon}$ by $|\boldsymbol{\varepsilon}|$  to denote the number of
positive $\varepsilon_j$'s in the cluster $\boldsymbol{\varepsilon}$.
Rearrange the $2^N$ terms in the $\tau$ function \eqref{tau-W2}
in terms of  $|\boldsymbol{\varepsilon}|$ such that
\begin{equation}\label{tau-epl}
\tau=\sum^{N}_{l=0} \sum_{|\boldsymbol{\varepsilon}|=l}\tau_{\boldsymbol{\varepsilon}}=\sum^{N}_{l=0}\tau^{(l)},
\end{equation}
where $\tau^{(l)}=\sum_{|\boldsymbol{\varepsilon}|=l}\tau_{\boldsymbol{\varepsilon}}$,
and in particular, by $g$ we denote $\tau^{(0)}$, i.e.
\begin{align}
g=\tau^{(0)}
=(-1)^{\frac{N(N-1)}{2} }
	\frac{\sigma(x-\sum_{i=1}^Nk_i)}{\sigma(x)}\cdot
\frac{\prod_{1\leq i<j\leq N}\sigma(k_i-k_j)}{\sigma^{N}(k_1)\cdots\sigma^{N}(k_N)}
\exp\left(\sum_{i=1}^N \gamma_i\right) . \label{g}
\end{align}

Here, for convenience of this  subsection, for a function $f=f(x)$,
by $\wt f$ we specially denote the $f$ shifted in $x$ by $\sum_{i=1}^Nk_i$, i.e.
$\wt f=f(x+\sum_{i=1}^Nk_i)$.
Then we have the following.

\begin{theorem}\label{T-2} 
Let
\begin{equation}\label{f-tg}
f=\frac{\widetilde \tau}{\wt g},
\end{equation}
where  $\tau$ and $g$ are given by \eqref{tau-W2} and \eqref{g}.
Then we have
\begin{equation}\label{bilinear-f}
 (D_x^4-4D_xD_t-12\wp(x)D_x^2)f\cdot f=0.
\end{equation}
$f$ is the $\tau$ function in Hirota's form, written as
\begin{equation}\label{f-Hirota}
f=\sum_{\mu=0,1} \frac{\sigma(x+2\sum_{i=1}^N \mu_i k_i)}{\sigma(x)\prod^N_{j=1}\sigma^{\mu_j}(2k_j)}
\mathrm{exp}\left(\sum^{N}_{j=1} \mu_j \theta_j+\sum^N_{1\leq i<j}\mu_i\mu_j a_{ij}\right),
\end{equation}
i.e.
\begin{equation*}
\begin{split}
 f=1&+ \sum_{i=1}^N\Phi_x(2k_i) e^{\theta_i}
	 +\mathop{\rm{\sum}}_{1\leq l<p\leq N} \frac{\sigma(x+2k_l+2k_p)}{\sigma(x)\sigma(2k_l)\sigma(2k_p)}
A_{l,p}e^{\theta_l+\theta_p}\\
	&+\cdots
	 +\frac{\sigma(x+2\sum_{i=1}^Nk_i)}{\sigma(x)\prod^N_{j=1}\sigma(2k_j)}
\left(\prod_{1\leq i<j\leq N}A_{i,j}\right)\prod_{i=1}^{N}e^{\theta_i},
    \end{split}
\end{equation*}
where the summation of $\mu$ means to take all possible $\mu_i=\{0,1\}$  for $ i=1,2,\cdots, N$,
\begin{subequations}\label{theta-A}
\begin{align}
& \theta_i=-2\zeta(k_i)x + \wp'(k_{i})t +\theta_{i}^{(0)}, ~~ \theta_{i}^{(0)}\in \mathbb{C},\label{theta}\\
& e^{a_{ij}}=A_{ij}=\left(\frac{\sigma(k_i-k_j)}{\sigma(k_i+k_j)}\right)^2.\label{A-ij}
\end{align}
\end{subequations}
\end{theorem}

\begin{proof}
The proof consists of two parts. First, we will prove that $f$ defined by \eqref{f-tg}
solves the bilinear KdV equation \eqref{bilinear-f}.
In the second part we will prove $f$ can be written into Hirota's form \eqref{f-Hirota}.

Recalling formulae \eqref{Dx2} and \eqref{Dx4}, for the function $g$ defined in \eqref{g}, we have
\begin{align*}
     &D_x^2~ \wt g\cdot \wt g=2\left(\wp(x+\hbox{$\sum_{i=1}^Nk_i$})-\wp(x)\right) \wt g^2,\\
	 &D_x^4~ \wt g\cdot \wt g=12\wp(x+\hbox{$\sum_{i=1}^Nk_i$}) D_x^2~ \wt g\cdot \wt g.
\end{align*}
Then, noticing that $\widetilde \tau=f\wt g$, and using the above relations and the quasi-gauge property
described in Proposition \ref{P-B1},
by calculation we find
\begin{equation*}
\begin{split}
0=&\, (D_x^4-4D_xD_t-12\wp(\hbox{$x+\sum_{i=1}^Nk_i)D_x^2$})\widetilde\tau\cdot \widetilde\tau\\
=&\, \wt g^2 (D_x^4-4D_xD_t-12\wp(x)D_x^2)f\cdot f,
\end{split} 
\end{equation*}
which indicates that $f=\widetilde \tau /\wt g$ solves the bilinear KdV equation \eqref{bilinear-f}.

In the second part, we are going to prove $f=\wt \tau/\wt g$ can be written as in \eqref{f-Hirota}.
In light of \eqref{tau-epl}, a generic term in $f$ is
\begin{equation}\label{f-tg-ep}
\frac{\wt \tau_{\boldsymbol{\varepsilon}}}{\wt g}
=\frac{\sigma(x+\sum_{j=1}^N (1+\varepsilon_j) k_j)}{\sigma(x)}\cdot
\left(\prod_{1\leq i<j\leq N}\frac{\varepsilon_i\sigma(\varepsilon_i k_i-\varepsilon_j k_j)}
{\sigma(k_i- k_j)}\right)
\exp \left(-\sum_{j=1}^N(1+\varepsilon_j)\wt\gamma_j\right).
\end{equation}
In particular, when $|\boldsymbol{\varepsilon}|=1$, e.g. only $\varepsilon_{j_0}=1$ while all other
$\varepsilon_{i}$'s are $-1$,
such a term is
\begin{equation*} 
\Phi_x(2k_{j_0}) e^{\alpha_{j_0}} e^{-2\wt \gamma_{j_0}},
\end{equation*}
where
\begin{equation*}
e^{\alpha_{j_0}}=\sigma(2k_{j_0}) \prod_{1\leq i \leq N \atop i\neq j_0}\frac{\sigma(k_{j_0}+ k_i)}
{\sigma(k_{j_0}- k_i)\cdot \mathrm{sgn}[i-j_0]}.
\end{equation*}
To proceed, we introduce
\[S=\{1,2,\cdots, N\},~~ J_{\boldsymbol{\varepsilon}}=\{n_1, n_2, \cdots, n_l\}\subset S,\]
where $J_{\boldsymbol{\varepsilon}}$ is associated with $\boldsymbol{\varepsilon}$ via
\[\varepsilon_i=\left\{\begin{array}{ll}
                       1, & i\in J_{\boldsymbol{\varepsilon}},\\
                       -1,& i\in S\setminus J_{\boldsymbol{\varepsilon}}.
                       \end{array}\right. \]
Eq.\eqref{f-tg-ep} is then written as
\begin{equation}\label{f-tg-ep-2}
\frac{\wt \tau_{\boldsymbol{\varepsilon}}}{\wt g}
=\frac{\sigma(x+2\sum_{i\in J_{\boldsymbol{\varepsilon}}} k_{i})}
{\sigma(x)\prod_{i\in J_{\boldsymbol{\varepsilon}}}\sigma(2k_{i})}
\left(\prod_{i\in J_{\boldsymbol{\varepsilon}}} e^{\beta_i}\right)
\exp \left(-2 \sum_{i\in J_{\boldsymbol{\varepsilon}}} \wt\gamma_i \right),
\end{equation}
where
\begin{equation*}
e^{\beta_i}=\sigma(2k_{i}) \prod_{j\in S\setminus J_{\boldsymbol{\varepsilon}}}
\frac{\sigma(k_{i}+ k_j)}
{\sigma(k_{i}- k_j)\cdot \mathrm{sgn}[j-i]}.
\end{equation*}
Then, noticing that
\[e^{\beta_i-\alpha_i}=\prod_{j\in J_{\boldsymbol{\varepsilon}}\atop j\neq i}
\frac{\sigma(k_{i}- k_j)}
{\sigma(k_{i}+ k_j)\cdot \mathrm{sgn}[j-i]},\]
which indicates that
\[\prod_{i\in J_{\boldsymbol{\varepsilon}}}e^{\beta_i-\alpha_i}
=\prod_{i,j\in J_{\boldsymbol{\varepsilon}}\atop  i<j}
\left(\frac{\sigma(k_{i}- k_j)}{\sigma(k_{i}+ k_j)}\right)^2
=\prod_{i,j\in J_{\boldsymbol{\varepsilon}}\atop  i<j} A_{ij},\]
where $A_{ij}$ is defined as in \eqref{A-ij},
the term \eqref{f-tg-ep-2} is written as
\begin{equation*}
\frac{\wt \tau_{\boldsymbol{\varepsilon}}}{\wt g}
=\frac{\sigma(x+2 \sum_{i\in J_{\boldsymbol{\varepsilon}}} k_i)}
{\sigma(x)\prod_{i\in J_{\boldsymbol{\varepsilon}}}\sigma(2k_{i})}\cdot
\left(\prod_{i,j\in J_{\boldsymbol{\varepsilon}}\atop  i<j} A_{ij} \right)
\exp \left( \sum_{i\in  J_{\boldsymbol{\varepsilon}}} \theta_i \right),
\end{equation*}
where $\theta_i=-2\wt\gamma_i+\alpha_i$, defined as in \eqref{theta}.
This indicates that $f=\wt\tau/\wt g$ can be written into Hirota's form \eqref{f-Hirota} coupled with \eqref{theta-A}.

The proof is completed.

\end{proof}

\subsection{Vertex operator}\label{sec-3-2}

We look for a vertex operator that generates the $\tau$ function  \eqref{f-Hirota} for elliptic solitons.
To proceed, let us first list some notations. Let $\mathbf{t}=(t_1=x,t_2, \cdots, t_n, \cdots)$,
$\widetilde{\partial}=(\partial_{t_1},\frac{1}{2} \partial_{t_2},\cdots,\frac{1}{n} \partial_{t_n},\cdots)$,
$\b{\mathbf{t}}=(t_1=x,t_3, \cdots, t_{2n+1}, \cdots)$,
$\b{\widetilde{\partial}}=(\partial_{t_1},\frac{1}{3} \partial_{t_3},\cdots,\frac{1}{(2n+1)} \partial_{t_{2n+1}},\cdots)$,
\begin{subequations}\label{xi-theta}
	\begin{align}
	&\xi (\mathbf{t},k)=\sum_{n=1}^{\infty}k^n t_n, ~~~
      \xi_{[e]}(\mathbf{t},k)=\sum_{n=1}^{\infty}(-1)^n\frac{\zeta^{(n-1)}(k)}{(n-1)!}t_n,~~ 
      \zeta^{(i)}(k)=\partial^i_k\zeta(k),\\
    & \theta(\b{\mathbf{t}},k)=\xi (\mathbf{t},k)-\xi (\mathbf{t},-k)
    =2\sum_{n=0}^{\infty} k^{2n+1} t_{2n+1},\\
    & \theta_{[e]}(\b{\mathbf{t}},k)=\xi_{[e]} (\mathbf{t},k)-\xi_{[e]} (\mathbf{t},-k)
    =-2\sum_{n=0}^{\infty}\frac{\zeta^{(2n)}(k)}{(2n)!} t_{2n+1}.
\end{align}
\end{subequations}
Consider the following $\tau$ function which is equivalent to \eqref{f-Hirota},
\begin{equation}\label{tau-ver}
\tau_N^{}(\b{\mathbf{t}})=\sum_{J\subset S}\left(\prod_{i\in J}c_i\right)\Biggl(\prod_{i,j\in J\atop i<j}A_{ij}\Biggr)
\frac{\sigma(t_1+2\sum_{i\in J}k_i)}{\sigma(t_1)\prod_{i\in J}\sigma(2k_i)}
\exp\left(\sum_{i\in J}\theta_{[e]}(\b{\mathbf{t}},k_i)\right),
\end{equation}
where $c_i$ are arbitrary constants, $A_{ij}$ is defined as in \eqref{A-ij},
$S=\{1,2,\cdots,N\}$, $J$ is a subset of $S$,
and $\sum_{J\subset S}$ means the summation runs over all subsets of $S$.
The vertex operator that generates the above $\tau$ function is described below.

\begin{theorem}\label{T-3}
The $\tau$ function \eqref{tau-ver} can be generated by the vertex operator
\begin{equation}\label{eq:vertexkdv}
X(k)=\Phi_{t_1}(2 k) e^{\theta_{[e]} (\b{\mathbf{t}},k)} e^{\theta (\b{\wt \partial},k)}
\end{equation}
via
\begin{equation}\label{tau-N-N}
 \tau_{N}(\b{\mathbf{t}})=e^{c_{N}X(k_{N})}\circ \tau_{N-1}(\b{\mathbf{t}}),~~
 \tau_{0}(\b{\mathbf{t}})=1,
\end{equation}
i.e.
\begin{equation}\label{tau-N}
\tau_N(\b{\mathbf{t}})=e^{c_{N}X(k_N)}\cdots e^{c_{2}X(k_2)}e^{c_{1}X(k_1)}\circ 1.
\end{equation}
In addition, $\tau_N(\b{\mathbf{t}})$ is doubly periodic with respect to any $k_i$,
for $i=1,2, \cdots, N$, where the two periods are those of $\wp(k)$.
\end{theorem}

Let us prove the theorem through the following lemmas.

\begin{lemma}\label{L-2}
For $\theta$ and $\theta_{[e]}$ defined in \eqref{xi-theta}, we have
\begin{equation}\label{X-exch}
e^{\theta (\b{\wt \partial},k_i)} e^{\theta_{[e]} (\b{\mathbf{t}},k_j)}=A_{ij} \, 
e^{\theta_{[e]} (\b{\mathbf{t}},k_j)} e^{\theta (\b{\wt \partial},k_i)},
\end{equation}
where
$A_{ij}$ is defined as \eqref{A-ij}.
\end{lemma}

\begin{proof}
Considering the Taylor series in the neighbourhood of $q=0$, we have
\begin{equation}\label{Aij-pq}
\ln \frac{\sigma(p-q)}{\sigma(p+q)}=\theta_{[e]}(\b\varepsilon(q), p),
\end{equation}
which indicates
\begin{equation}
A_{ij}=e^{2\theta_{[e]}(\b\varepsilon(k_j), k_i)},
\end{equation}
where $\b\varepsilon(q)=(q,\frac{q^3}{3},\cdots,\frac{q^{2n+1}}{(2n+1)},\cdots)$.
Then, for any $C^{\infty}$ function $h(\b{\mathbf{t}})$, one can directly verify that
\[e^{\theta (\b{\wt \partial},k_i)} e^{\theta_{[e]} (\b{\mathbf{t}},k_j)}\circ h(\b{\mathbf{t}})
=A_{ij} \, e^{\theta_{[e]} (\b{\mathbf{t}},k_j)} e^{\theta (\b{\wt \partial},k_i)}\circ h(\b{\mathbf{t}}),
\]
i.e. relation \eqref{X-exch} holds.
\end{proof}

Note that \eqref{X-exch} is formally similar to the result in the usual soliton case, cf.\cite{DKM-PJA-1981,MJD-book-1999}.
We are led by this lemma to the following.

\begin{lemma}\label{L-3}
For the vertex operator  $X(k)$ defined by \eqref{eq:vertexkdv}, we have
\begin{subequations}
\begin{align}
	&X(k_i)X(k_j)=A_{i,j}\frac{\sigma(t_1+2k_i+2 k_j)}{\sigma(t_1)\sigma(2k_i)\sigma(2k_j)}
      e^{\theta_{[e]}(\b{\mathbf{t}},k_i,)+\theta_{[e]}(\b{\mathbf{t}},k_j)}
      e^{\theta(\b{\wt{\partial}}, k_i)+\theta(\b{\wt \partial}, k_j)},\label{XX-commu}\\
	&X(k_s)\cdots X(k_2)X(k_1) \nonumber \\
	&\qquad=\left(\prod_{1\leq i<j\leq s}A_{ij}\right)
       \frac{\sigma(t_1+2\sum_{i=1}^sk_i)}{\sigma(t_1)\prod_{i=1}^{s}\sigma(2k_i)}
       \exp\left( \sum_{i=1}^s\theta_{[e]}(\b{\mathbf{t}},k_i)\right)\cdot
       \exp\left(\sum_{i=1}^s\theta(\b{\wt{\partial}}, k_i)\right),
\end{align}
\end{subequations}
and hence
\begin{equation}
X(k)^2=0,~~
e^{cX(k)}=1+cX(k).
\end{equation}
\end{lemma}

With the above two lemmas in hand, we can confirm that $\tau_N(\b{\mathbf{t}})$
can be generated by the vertex operator $X(k)$ via \eqref{tau-N}, with \eqref{tau-N-N} as a consequence.
In addition, noticing that $\tau_1(\b{\mathbf{t}})= e^{c_{1}X(k_1)}\circ 1$ is doubly periodic with respect to $k_1$,
and $X(k_i)$ and $X(k_j)$ commute (see  \eqref{XX-commu}
where we should consider $A_{ij}$ to be a rational function rather than a Laurent series of $k_i/k_j$ or $k_j/k_i$,
cf. \cite{DKM-PJA-1981}),
it  follows that $\tau_N(\b{\mathbf{t}})$ defined by \eqref{tau-N} is
doubly periodic with respect to any $k_i$, for $i=1,2, \cdots, N$.
Thus Theorem \ref{T-3} holds.

\subsection{Bilinear identity of the KdV hierarchy}\label{sec-3-3}

With the vertex operator and $\tau$ function in hand, we can have bilinear forms
of the KdV hierarchy that admit elliptic soliton solutions.

To achieve that, let us first introduce a doubly periodic function.

\begin{lemma}\label{L-4}
Consider a vertex operator
\begin{align}
	 X(\b{\mathbf{t}},q)=\frac{\sigma(t_1+q)}{\sigma(q)} e^{\frac{1}{2}\theta_{[e]}(\b{\mathbf{t}},q)}
     e^{\frac{1}{2}\theta(\b{\wt{\mathbf{\partial}}},q)},
	\end{align}
and introduce a function of $q$,
\begin{equation}\label{h-q}
	 h (\b{\mathbf{t}},q)= X(\b{\mathbf{t}},q)\tau (\b{\mathbf{t}}),
\end{equation}
where $\tau (\b{\mathbf{t}})=\tau_N^{} (\b{\mathbf{t}})$ is defined by \eqref{tau-ver}.
Then, $ h (\b{\mathbf{t}},q)$ is a doubly periodic function of $q$
with periods $2w_1$ and $2w_2$, where $w_i$ are the half periods of $\wp(q)$.
\end{lemma}

\begin{proof}
Making use of relation \eqref{Aij-pq}, $h (\b{\mathbf{t}},q)$ can be explicitly written as
\begin{align}
h(\b{\mathbf{t}},q)=&\frac{\sigma(t_1+q)}{\sigma(q)}  e^{\frac{1}{2}\theta_{[e]}(\b{\mathbf{t}},q)}\nonumber \\
	& \times	\sum_{J\subset S}\left[
       \!\!\left(\prod_{i<j\in J}A_{ij}\right)
       \frac{\sigma(t_1+2\sum_{i\in J}k_i+q)}{\sigma(t_1+q)\prod_{i\in J} \sigma(2k_i)}
		\left(\prod_{i\in J}c_i\frac{\sigma(k_i-q)}{\sigma(q+k_i)} \right)
        \times e^{\sum_{i\in J}\theta_{[e]}(\b{\mathbf{t}},k_i)}\right]. \label{h-qq}
\end{align}
Note that in $\frac{1}{2}\theta_{[e]}(\b{\mathbf{t}},q)$, except the first term $-\zeta(q)t_1$,  the rest part
$-\sum_{n=1}^{\infty}\frac{\zeta^{(2n)}(q)}{(2n)!} t_{2n+1}$ is already doubly periodic with respect to $q$.
Following Proposition \ref{P-A1}, one can check that
\[\frac{\sigma(t_1+q)}{\sigma(q)}  e^{-\zeta(q)t_1},~~~
 \frac{\sigma(t_1+2\sum_{i\in J}k_i+q)}{\sigma(t_1+q)\prod_{i\in J} \sigma(2k_i)}
		\prod_{i\in J}c_i\frac{\sigma(k_i-q)}{\sigma(q+k_i)}
\]
are doubly periodic too. This indicates $h(\b{\mathbf{t}},q)$ is a doubly periodic function of $q$.
Note that $h(\b{\mathbf{t}},q)$ is not elliptic as it has an essential singularity $q=0$
(mod periodic lattice).

\end{proof}

Then we come up with an integral bilinear identity.

\begin{theorem}\label{T-4}
For the function $h(\b{\mathbf{t}},q)$ defined in \eqref{h-q}, we have the following bilinear identity
\begin{equation}\label{bil-int-KdV}
\oint_{\Omega} \frac{\mathrm{d}q}{2\pi i}\, h(\b{\mathbf{t}},q)\, h({\mathbf{\b t}}',-q) =0,
\end{equation}
which gives rise to
\begin{equation}\label{bil-res-KdV}
\underset{q=0}{\mathrm{Res}}\left[ h(\b{\mathbf{t}},q)\, h({\mathbf{\b t}}',-q)\right]=0,
\end{equation}
where the  contour $\Omega$ takes the boundary, anticlockwise, of the open fundamental period parallelogram $\mathbb{D}$
(see Fig.\ref{Fig-1}) and all $\{\pm k_i\}$ are distinct and belong to $\mathbb{D}$.
\end{theorem}

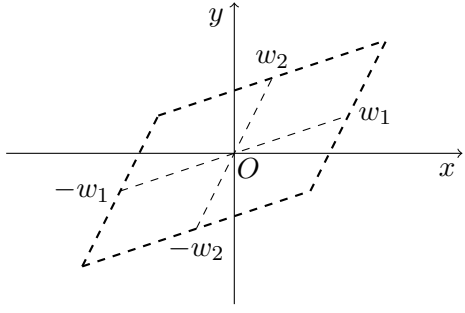
\begin{figure}[h]
\begin{center}
\begin{tikzpicture}[rotate=0]
\draw[->](-3,0)--(3,0);
\draw[->](0,-2)--(0,2);
\node[below] at (2.8,0) {$x$};
\node[left] at (0,1.8) {$y$};
\draw[dashed](-1.5,-0.5)--(1.5,0.5);
\draw[dashed,thick](-1,0.5)--(2.0,1.5);
\draw[dashed,thick](-2.0,-1.5)--(1.0,-0.5);
\draw[dashed](-0.5,-1.0)--(0.5,1.0);
\draw[dashed,thick](1.0,-0.5)--(2.0,1.5);
\draw[dashed,thick](-2,-1.5)--(-1,0.5);
\node[above] at (0.5,1.0) {$w_2$};
\node[below] at (-0.5,-1.0) {$-w_2$};
\node[left] at (-1.5,-0.5) {$-w_1$};
\node[right] at (1.5,0.5) {$w_1$};
\node[right] at (-0.1,-0.2) {$O$};
\end{tikzpicture}
\end{center}
\vskip -0.5cm
\caption{Fundamental period parallelogram $\mathbb{D}$}\label{Fig-1}
\end{figure}

\begin{proof}
In light of Lemma \ref{L-4}, it is obvious the integrand $h(\b{\mathbf{t}},q) h({\mathbf{\b t}}',-q)$
is a double-periodic function of $q$.
Meanwhile, noticing that in $\mathbb{D}$ the integrand has only $2N$ isolated simple poles $\{\pm k_i\}_{i=1}^{N}$
and one isolated essential singularity $q=0$,
there is a domain which contains the curve $\Omega$ and where the integrand is continuous.
It then turns out that the integral in \eqref{bil-int-KdV} is zero due to the integrand being double-periodic.

To prove the second identity \eqref{bil-res-KdV}, we examine residues of the integrand at $q=\pm k_i$.
For given $j_0\in S$, $q=k_{j_0}$ is a simple pole of $h({\mathbf{\b t}}',-q)$
but $h(\b{\mathbf{t}},q)$ is analytic at this point. Thus we have
\begin{equation}\label{res-hh}
\underset{q=k_{j_0}}{\mathrm{Res}}\left[ h(\b{\mathbf{t}},q)\, h({\mathbf{\b t}}',-q)\right]
=h(\b{\mathbf{t}},k_{j_0}) \times \underset{q=k_{j_0}}{\mathrm{Res}}\left[h({\mathbf{\b t}}',-q)\right].
\end{equation}
$h({\mathbf{\b t}}',-q)$ has a similar summation expression as \eqref{h-qq}.
For any $J$ that does not contain $j_0$, the associated terms
in the summation expression of $h({\mathbf{\b t}}',-q)$ contribute nothing
to the residue at $q=k_{j_0}$.
Therefore we have
\begin{equation}
\underset{q=k_{j_0}}{\mathrm{Res}}\left[ h({\mathbf{\b t}}',-q)\right]
=\underset{q=k_{j_0}}{\mathrm{Res}}\left[g({\mathbf{\b t}}',q)\right],
\end{equation}
where $g(\b{\mathbf{t}}',q)$ is a collection of all those $k_{j_0}$-related terms in $h({\mathbf{\b t}}',-q)$,
which is
\begin{align*}
g(\b{\mathbf{t}}',q)=&\frac{- e^{\frac{1}{2}\theta_{[e]}(\b{\mathbf{t}}',-q)}}{\sigma(q)}  
	\sum_{J\subset S\backslash\{j_0\}}\left[%\left(\prod_{i\in J}c_i\right)
       \!\!\left(\prod_{i<j\in J}A_{ij}\right)
       \left(\prod_{i\in J}c_i\frac{\sigma(k_i+q)}{\sigma(k_i- q)} \right)
        \times\frac{ e^{\sum_{i\in J}\theta_{[e]}(\b{\mathbf{t}}',k_i)}}{\prod_{i\in J}\sigma(2k_i)}\times B_{j_0}\right]
\end{align*}
where
\[B_{j_0}=c_{j_0} \sigma(t_1'+2\hbox{$\sum_{i\in J}k_i$}+2k_{j_0}-q)
\left(\prod_{i\in J}\frac{\sigma^2(k_i-k_{j_0})}{\sigma^2(k_i+k_{j_0})}\right)
\frac{\sigma(k_{j_0}+q)}{\sigma(k_{j_0}- q)}\cdot
 \frac{e^{\theta_{[e]}(\b{\mathbf{t}}',k_{j_0})}}{\sigma(2k_{j_0})}.
\]
Note that $q=k_{j_0}$ is a simple pole of $B_{j_0}$. A direct calculation yields
\begin{align*}
& \underset{q=k_{j_0}}{\mathrm{Res}}\left[g(\b{\mathbf{t}}',q)\right]\\
=~&\frac{ c_{j_0}   e^{\frac{1}{2}\theta_{[e]}(\b{\mathbf{t}}',k_{j_0})}}{\sigma(k_{j_0})}  \\
	& \times	\sum_{J\subset S\backslash\{j_0\}}\left[%\left(\prod_{i\in J}c_i\right)
       \!\!\left(\prod_{i<j\in J}A_{ij}\right)
       \left(\prod_{i\in J}c_i\frac{\sigma(k_i-k_{j_0})}{\sigma(k_i+k_{j_0})} \right)
        \frac{\sigma(t'_1+2\hbox{$\sum_{i\in J}k_i$}+k_{j_0})}{\prod_{i\in J}\sigma(2k_i)}
        \cdot e^{\sum_{i\in J}\theta_{[e]}(\b{\mathbf{t}}',k_i)}\right],
\end{align*}
where we have made use of
\[\lim_{q\to k_{j_0}}\left(\prod_{i\in J}c_i\frac{\sigma(k_i+q)}{\sigma(k_i- q)} \right)
\left(\prod_{i\in J}\frac{\sigma^2(k_i-k_{j_0})}{\sigma^2(k_i+k_{j_0})}\right)
=\prod_{i\in J}c_i\frac{\sigma(k_i-k_{j_0})}{\sigma(k_i+k_{j_0})}.
\]
Recalling the expression \eqref{h-qq} for $h(\b{\mathbf{t}},q)$,
in the summation, such terms will vanish as they are generated by $J$ that contains $j_0$.
Thus from \eqref{res-hh} we have
\begin{align*}
&\underset{q=k_{j_0}}{\mathrm{Res}}\left[h(\b{\mathbf{t}},q) h(\b{\mathbf{t}}',-q)\right]\\
=~&\frac{ c_{j_0}
e^{\frac{1}{2}\theta_{[e]}(\b{\mathbf{t}},k_{j_0})+\frac{1}{2}\theta_{[e]}(\b{\mathbf{t}}',k_{j_0})}}
{\sigma^2(k_{j_0})} \\
	& \times	\sum_{J\subset S\backslash\{j_0\}}\left[%\left(\prod_{i\in J}c_i\right)
       \!\!\left(\prod_{i<j\in J}A_{ij}\right)
       \left(\prod_{i\in J}c_i\frac{\sigma(k_i-k_{j_0})}{\sigma(k_i+k_{j_0})} \right)
        \frac{\sigma(t_1+2\hbox{$\sum_{i\in J}k_i$}+k_{j_0})}{\prod_{i\in J}\sigma(2k_i)}
        \cdot e^{\sum_{i\in J}\theta_{[e]}(\b{\mathbf{t}},k_i)}\right]\\
     & \times	\sum_{J\subset S\backslash\{j_0\}}\left[
       \!\!\left(\prod_{i<j\in J}A_{ij}\right)
       \left(\prod_{i\in J}c_i\frac{\sigma(k_i-k_{j_0})}{\sigma(k_i+k_{j_0})} \right)
        \frac{\sigma(t_1'+2\hbox{$\sum_{i\in J}k_i$}+k_{j_0})}{\prod_{i\in J}\sigma(2k_i)}
        \cdot e^{\sum_{i\in J}\theta_{[e]}(\b{\mathbf{t}}',k_i)}\right]   .
\end{align*}
In a similar way we can calculate the residue of the integrand at $q=-k_{j_0}$.
It turns out that
\[\underset{q=-k_{j_0}}{\mathrm{Res}}\left[h(\b{\mathbf{t}},q) h(\b{\mathbf{t}}',-q)\right]
=- \underset{q=k_{j_0}}{\mathrm{Res}}\left[h(\b{\mathbf{t}},q) h(\b{\mathbf{t}}',-q)\right],
\]
which means finally all residues at $q=\pm k_i$ cancel,
and the remained residue at $q=0$ gives rise to the bilinear identity \eqref{bil-res-KdV}.

The proof is completed.
\end{proof}

\subsection{Algorithm for calculating residues}\label{sec-3-4}

In the following we formulate an algorithm to calculate residues from the identity \eqref{bil-res-KdV}
so that the bilinear KdV hierarchy with elliptic solitons can be obtained.

Redefining
$\tau'(\b{\mathbf{t}})=\sigma(t_1) \tau(\b{\mathbf{t}})$,
 the bilinear identity \eqref{bil-int-KdV} is written as
\begin{equation}\label{bil-int-KdVV}
\oint_{\Omega} \frac{\mathrm{d}q}{2\pi i}\, \frac{1}{\sigma^2(q)}
e^{\frac{1}{2}\theta_{[e]}(\b{\mathbf{t}}-\b{\mathbf{t}}',q)}
\tau'(\b{\mathbf{t}}+\b\varepsilon(q))  \tau'(\b{\mathbf{t}}'-\b\varepsilon(q)) =0.
\end{equation}
Then, introducing
$\b{\mathbf{t}}=\b{\mathbf{x}}+\b{\mathbf{y}}$ and $\b{\mathbf{t'}}=\b{\mathbf{x}}-\b{\mathbf{y}}$,
where $\b{\mathbf{x}}=(x_1, x_3, \cdots)$, $\b{\mathbf{y}}=(y_1, y_3, \cdots)$,
the above equation  is written as
\begin{equation}
\oint_{\Omega} \frac{\mathrm{d}q}{2\pi i}\, \frac{1}{\sigma^2(q)} e^{\theta_{[e]}(\b{\mathbf{y}},q)}
e^{(\b{\mathbf{y}}+\b\varepsilon(q))\cdot \mathbf{D}_{\b{\mathbf{x}}}}
\tau'(\b{\mathbf{x}}) \cdot \tau'(\b{\mathbf{x}}) =0,
\end{equation}
and from \eqref{bil-res-KdV} we have
\begin{equation}\label{bil-D-KdV}
\underset{q=0}{\mathrm{Res}}\left[ \frac{1}{\sigma^2(q)} e^{\theta_{[e]}(\b{\mathbf{y}},q)}
e^{(\b{\mathbf{y}}+\b\varepsilon(q))\cdot \mathbf{D}_{\b{\mathbf{x}}}}
\tau'(\b{\mathbf{x}}) \cdot \tau'(\b{\mathbf{x}})\right] =0,
\end{equation}
where $\mathbf{D}_{\b{\mathbf{x}}}=(D_{x_1}, D_{x_3}, D_{x_5},\cdots)$,
and for two vectors $\mathbf{a}=(a_1, a_2, \cdots)$ and $\mathbf{b}=(b_1, b_2, \cdots)$
their vector product is defined as $\mathbf{a}\cdot \mathbf{b}=\sum_{i=1} a_ib_i$.
Note that in the usual soliton case, the term $\frac{1}{\sigma^2(q)} e^{\theta_{[e]}(\b{\mathbf{y}},q)}$
in \eqref{bil-D-KdV}
is $\frac{1}{q^2} e^{\theta(\b{\mathbf{y}},1/q)}$ instead, cf.\cite{JM-RIMS-1983,MJD-book-1999};
$e^{\theta(\b{\mathbf{y}},1/q)}$ has a definite expansion in terms of $q$
but $e^{\theta_{[e]}(\b{\mathbf{y}},q)}$ does not.\footnote{One can formally write
$e^{\theta_{[e]}(\b{\mathbf{y}},q)}=\sum_{j=-\infty}^{\infty}h_j(\b{\mathbf{y}}) q^j$
but $h_j(\b{\mathbf{y}})$  can not be expressed explicitly.}
This is the obstacle when calculating the residue at $q=0$.
We need to design an algorithm to calculate the residue in \eqref{bil-D-KdV}.

To develop the algorithm we write \eqref{bil-D-KdV} into the following form
\begin{equation}\label{bil-D-KdV-1}
\underset{q=0}{\mathrm{Res}}\left[ e^{(\mathbf{\b B}+\mathbf{D}_{\b{\mathbf{x}}})\cdot \b{\mathbf{y}}}
\frac{1}{\sigma^2(q)}
e^{\xi(\wt{\mathbf{D}}_{\b{\mathbf{x}}},\,q)}
\tau'(\b{\mathbf{x}}) \cdot \tau'(\b{\mathbf{x}})\right] =0,
\end{equation}
where $\xi$ is defined as in \eqref{xi-theta} and
\[\mathbf{\b B}=-2(\zeta(q), \frac{\zeta''(q)}{2!}, \cdots \frac{\zeta^{(2n)}(q)}{(2n)!}, \cdots),~~
\wt{\mathbf{D}}_{\b{\mathbf{x}}}=(D_{x_1},0 ,\frac{1}{3}D_{x_3},0,\frac{1}{5}D_{x_5}, \cdots).
\]
For convenience, we introduce polynomials $\{p_n(\mathbf{t})\}$ by \cite{MJD-book-1999}
\begin{equation}\label{p-j}
e^{\xi(\mathbf{t},k)}=\sum_{n=0}^{\infty}p_n(\mathbf{t})k^n,
\end{equation}
where
\begin{align*}
  &p_n(\mathbf{t})=\sum_{\|\alpha\|=n}\frac{\mathbf{t}^\alpha}{\alpha!},\\
  &\mathbf{t}=(t_1,t_2,\cdots),\quad \alpha=(\alpha_1,\alpha_2,\cdots),\\
  &\|\alpha\|=\sum_{j=0}^{\infty}j\alpha_j,\quad \alpha!=\alpha_1!\alpha_2!\cdots,
  \quad \mathbf{t}^\alpha=t_1^{\alpha_1}t_2^{\alpha_2}\cdots.
\end{align*}
The first few $\{p_n(\mathbf{t})\}$'s are
\begin{align*}
& p_0(\mathbf{t})=1,~~ p_1(\mathbf{t})=t_1,~~ p_2(\mathbf{t})=\frac{1}{2}t^2_1+t_2,\\
& p_3(\mathbf{t})=\frac{1}{3!}t^3_1+t_1t_2+t_3,~~
p_4(\mathbf{t})=\frac{1}{4!}t^4_1+\frac{1}{2}t_1^2t_2+\frac{1}{2}t_2^2+t_1t_3+t_4.
\end{align*}
Meanwhile, $1/\sigma^2(q)$ is expanded as
\begin{equation}\label{sigma-expd}
\frac{1}{\sigma^2(q)}=\sum^{\infty}_{j=0}\mu_j q^{j-2}
=\frac{1}{q^2}(1+\frac{g_2}{120}q^4+\frac{g_3}{420}q^6+\frac{13g_2^2}{201600}q^8+\cdots).
\end{equation}
Then, the bilinear identity \eqref{bil-D-KdV-1} is written as
\begin{align}
0&=\underset{q=0}{\mathrm{Res}}\left[
\left(\sum_{|\b \beta|=0}^{\infty}\frac{(\b{\mathbf{B}}
+\mathbf{D}_{\b{\mathbf{x}}})^{\b \beta}}{\b \beta !}\b{\mathbf{y}}^{\b \beta}\right)
\left(\sum_{n=0}^{\infty}\sum_{j=0}^n p_j(\wt{\mathbf{D}}_{\b{\mathbf{x}}})\mu_{n-j} q^{n-2}\right)
\tau'(\b{\mathbf{x}}) \cdot \tau'(\b{\mathbf{x}})\right] \nonumber \\
&=\sum_{|\b \beta|=0}^{\infty}\underset{q=0}{\mathrm{Res}}\left[
\frac{(\mathbf{\b B}+\mathbf{D}_{\b{\mathbf{x}}})^{\b \beta}}{\b \beta !}
\left( \sum_{n=0}^{\infty}\sum_{j=0}^n p_j(\wt{\mathbf{D}}_{\b{\mathbf{x}}})\mu_{n-j} q^{n-2}\right)
\tau'(\b{\mathbf{x}}) \cdot \tau'(\b{\mathbf{x}})\right]\b{\mathbf{y}}^{\b \beta},
\label{bil-D-KdV-2}
\end{align}
where $\b \beta=(\beta_1, \beta_3, \cdots, \beta_{2j+1}, \cdots)$ and $|\b \beta|=\sum_{j=0}^{\infty}\beta_{2j+1}$.
Since $\{y_i\}$ are arbitrary, it then follows that
\begin{equation}\label{bil-D-KdV-3}
\underset{q=0}{\mathrm{Res}}\left[
(\mathbf{\b B}+\mathbf{D}_{\b{\mathbf{x}}})^{\b \beta}
\left( \sum_{n=0}^{\infty}\sum_{j=0}^n p_j(\wt{\mathbf{D}}_{\b{\mathbf{x}}})\mu_{n-j} q^{n-2}\right)
\tau'(\b{\mathbf{x}}) \cdot \tau'(\b{\mathbf{x}})\right]=0.
\end{equation}
In the above equation, $\sum_{n=0}^{\infty}\sum_{j=0}^n p_j(\wt{\mathbf{D}}_{\b{\mathbf{x}}})\mu_{n-j} q^{n-2}$
is a Laurent series of $q$ starting from $q^{-2}$.
For another term $(\mathbf{\b B}+\mathbf{D}_{\b{\mathbf{x}}})^{\b \beta}$,
first, given $\b\beta$, % with finite $|\b \beta|$, $\b\beta$
contains only finite number of nonzero $\beta_j$.
Thus assume  $\b \beta=(\beta_1, \beta_3, \cdots, \beta_{2n+1},$ $ 0, 0, \cdots)$ without loss of generality.
Meanwhile, we shall note that the entries in $\mathbf{\b B}$ have a form $B_{2j+1}=-2\frac{\zeta^{(2j)}(q)}{(2j)!}$
where $\zeta(q)$ can be expanded as \eqref{zeta-expd}.
Since $|\b \beta|$ is finite,
$(\mathbf{\b B}+\mathbf{D}_{\b{\mathbf{x}}})^{\b \beta}$ is a Laurent series of $q$ as well
and it starts from $q^{-||\b \beta||}$ where $||\b \beta||=\sum_{j=0}^n(2j+1)\beta_{2j+1}$ is finite and positive.
This means, to calculate the residue \eqref{bil-D-KdV-3}, it is sufficient to consider the finite number of terms from
$q^{-||\b \beta||}$ to $q^1$ in $(\mathbf{\b B}+\mathbf{D}_{\b{\mathbf{x}}})^{\b \beta}$
and the finite number of terms from $q^{-2}$ to $q^{||\b \beta||-1}$ in
$\sum_{n=0}^{\infty}\sum_{j=0}^n p_j(\wt{\mathbf{D}}_{\b{\mathbf{x}}})\mu_{n-j} q^{n-2}$.
Thus, we are led to the following theorem which formulates an algorithm to derive bilinear KdV hierarchy
through calculating residues \eqref{bil-D-KdV-3}.

\begin{theorem}\label{T-5}
The bilinear KdV hierarchy are given by
\begin{equation}\label{bil-D-KdV-4}
\underset{q=0}{\mathrm{Res}}\left[
(\mathbf{\b B}+\mathbf{D}_{\b{\mathbf{x}}})^{\b \beta}|_{\leq 1}
\left( \sum_{n=0}^{||\b \beta||-1}\sum_{j=0}^n p_j(\wt{\mathbf{D}}_{\b{\mathbf{x}}})\mu_{n-j} q^{n-2}\right)
\tau'(\b{\mathbf{x}}) \cdot \tau'(\b{\mathbf{x}})\right]=0,
\end{equation}
where $\b \beta$ is set of nonnegative integers with finite and positive $|\b \beta|$,
and $(\mathbf{\b B}+\mathbf{D}_{\b{\mathbf{x}}})^{\b \beta}|_{\leq 1}$
means those  terms of $q^j$ with $j\leq 1$ in the Laurent series of
$(\mathbf{\b B}+\mathbf{D}_{\b{\mathbf{x}}})^{\b \beta}$.
\end{theorem}

As examples, when $\b \beta=(3,0,0,\cdots)$, from the above theorem we find
\[(D_{x_1}^4-4 D_{x_1} D_{x_3}-g_2) \tau'\cdot \tau'=0,
\]
which is the bilinear KdV equation \eqref{bilinear-aa}.
For the cases $\b \beta=(2,1,0,\cdots)$ and $\b \beta=(5,0,0,\cdots)$, we have, respectively,
\begin{subequations}
\begin{align}
    &(D_{x_1}^6+4D_{x_1}^3D_{x_3}-32D_{x_3}^2+3g_2D_{x_1}^2-24g_3) \tau'\cdot \tau'=0, \label{bl-2}\\
    &(D_{x_1}^6 + 40 D_{x_1}^3 D_{x_3} + 40 D_{x_3}^2 - 216 D_{x_1} D_{x_5} + 3 g_2 D_{x_1}^2 - 24 g_3)
    \tau'\cdot \tau'=0. \label{bl-1}
	\end{align}
\end{subequations}
When $g_2$, $g_3$ are 0, these equations degenerate to those in the KdV hierarchy for the usual soliton case, 
cf.\cite{JM-RIMS-1983}.

\section{$\tau$ function, vertex operator and bilinear identity: KP}\label{sec-4}

Both the KdV and KP equation serve as representative models in integrable systems,
while the latter plays a more fundamental role in Sato's theory of integrable systems.
Based on the exploration in the previous section for the KdV equation,
in this section we will focus on the KP equation and investigate its $\tau$ function, vertex operator
and bilinear identity associated with elliptic solitons.

\subsection{Elliptic $N$-solitons and $\tau$ function in Hirota's form}\label{sec-4-1}

The KP equation is\footnote{Usually,
\[4u_t-u_{xxx}-6uu_x-3\alpha^2\partial^{-1}u_{yy}=0\]
is known as KP-I when $\alpha^2=-1$ and  KP-II when $\alpha^2=1$.
We consider KP-II without loss of generality.}
\begin{equation}\label{eq:kp}
	4u_t-u_{xxx}-6uu_x-3\partial^{-1}u_{yy}=0,
\end{equation}
or in the potential form $(u=v_x)$
\begin{equation}\label{eq:kp-v}
	4v_t-v_{xxx}-3(v_x)^2-3\partial^{-1}v_{yy}=0.
\end{equation}
By the transformation
\begin{align}
u=-2\wp(x)+2(\ln \tau)_{xx},
\end{align}
or
\begin{equation}
v=2\zeta(x)+\frac{g_2}{4} t+2(\ln\tau)_x,
\end{equation}
the KP equation is bilinearised as
\begin{equation}\label{bilinear-b}
(D_x^4- 4D_xD_t-12\wp(x)D_x^2+3D_y^2)\tau\cdot\tau=0,
\end{equation}
or
\begin{equation}
	(D_x^4-4D_xD_t+3 D_y^2-g_2)\tau'\cdot\tau'=0,
\end{equation}
where $\tau'=\sigma(x)\tau$. The bilinear KP equation allows elliptic soliton solutions.

\begin{theorem}\label{T-6}
The following Wronskian
\begin{equation}\label{slt:KPWronskian}
	\tau=|\widehat{N-1}|
\end{equation}
is a solution to the bilinear KP equation \eqref{bilinear-b}, where $\tau$ is composed by vector
$\varphi=(\varphi_1, \cdots, \varphi_N)^T$
with entries
\begin{subequations}
\begin{equation}
\varphi_j(x,y,t) = \Phi_x(k_j)e^{-\gamma(k_j)}+ \Phi_x(l_j)e^{-\gamma(l_j)},
\end{equation}
where
\begin{equation}
\gamma(k)=\zeta(k)x+ \wp(k)y-\frac{\wp'(k)}{2}t+\gamma^{(0)}(k),~~ k\in \mathbb{C}
\end{equation}
with a constant $\gamma^{(0)}(k)$ related to $k$.
\end{subequations}
Note that $\varphi_j$ satisfies
\begin{equation}\label{eq:kplax}
	\begin{split}
	&\varphi_{j,y}= -\varphi_{j,xx}+2\wp(x)\varphi_j,\\
	&\varphi_{j,t}=\varphi_{j,xxx}-3\wp(x)\varphi_{j,x}-\frac{3}{2}\wp'(x)\varphi_j.
	\end{split}
\end{equation}
\end{theorem}

The proof will be given in Appendix \ref{A-3}.

To find out a corresponding Hirota's form of the $\tau$ function \eqref{slt:KPWronskian},
we consider \eqref{slt:KPWronskian} to be a summation of $2^N$ terms,
i.e. $\tau=\sum_{J \subset S}\tau_{J}^{}$,
where the generic term $\tau_{J}^{}$ is the Wronskian $|\widehat{N-1}|$ generated by
\begin{equation}\label{vphi-J}
\varphi=(\phi_1, \phi_2, \cdots, \phi_N)^T,
\end{equation}
in which $\phi_j=\Phi_x(k_j)e^{-\gamma(k_j)}$ for $j\in J$
and $\phi_j=\Phi_x(l_j) e^{-\gamma(l_j)}$ for $j\in S\backslash J$,
$J$ is a subset of $S=\{1,2,\cdots,N\}$.
In light of Lemma \ref{L-1}, we immediately get the following result.

\begin{lemma}\label{L-5}
The Wronskian $\tau_{J}^{}$ generated by vector \eqref{vphi-J} can be expressed as
\begin{align}
\tau_{J}^{}=&(-1)^{\frac{N(N-1)}{2}}
\frac{\sigma(x+\sum_{i\in J }k_{i}+\sum_{j\in S\backslash J}l_j)}{\sigma(x)}
\frac{\prod_{i\in J \atop j\in S\backslash J}\sigma(k_i-l_j)\mathrm{sgn}[j-i]}
{\left(\prod_{i\in J}\sigma^N(k_{i})\right)\left(\prod_{ j\in S\backslash J}\sigma^N(l_{i})\right)} \nonumber\\
 & \times \left(\prod_{i<j\in J}\sigma(k_i-k_j) \right)
 \left(\prod_{i<j\in S\backslash J}\sigma(l_{i}-l_j)\right)
 \exp\left[-\sum_{i\in J}\gamma(k_i)-\sum_{j\in S\backslash J}\gamma(l_j)\right],
\end{align}
especially, when $J$ is the empty set $\varnothing$, we have
\begin{equation}\label{g-KP}
g(x,y,t)=\tau_{\varnothing}^{}
=(-1)^{\frac{N(N-1)}{2}}
\frac{\sigma(x+\sum_{j\in S}l_j)}{\sigma(x)}
\frac{\prod_{i<j \in S}\sigma(l_i-l_j)}
{\prod_{ j\in S}\sigma^N(l_{i})} \mathrm{exp}\left(-\sum_{j\in S}\gamma(l_j)\right).
\end{equation}
\end{lemma}

Next, for a function $f(x)$, we introduce notation $\wt f(x)=f(x-\sum_{j=1}^N l_j)$.
Then, similar to the KdV case, we have the following.

\begin{theorem}\label{T-7}
For the function $\tau$ in Wronskian form \eqref{slt:KPWronskian} and $g$  given by \eqref{g-KP},
\begin{equation}\label{f-tg-KP}
f=\frac{\widetilde \tau}{\wt g}
\end{equation}
is a solution to the bilinear KP equation \eqref{bilinear-b}, i.e.
\begin{equation}\label{bilinear-bf}
(D_x^4- 4D_xD_t-12\wp(x)D_x^2+3D_y^2)f \cdot f=0,
\end{equation}
and $f$ is written in Hirota's form as
\begin{equation}\label{f-Hirota-KP}
f=\sum_{\mu=0,1} \frac{\sigma(x+\sum_{i=1}^N \mu_i (k_i-l_i))}{\sigma(x)\prod^N_{i=1}\sigma^{\mu_i}(k_i-l_i)}
\mathrm{exp}\left(\sum^{N}_{j=1} \mu_j \theta_j+\sum^N_{1\leq i<j}\mu_i\mu_j a_{ij}\right),
\end{equation}
where the summation of $\mu$ means to take all possible $\mu_i=\{0,1\}$  for $ i=1,2,\cdots, N$,
\begin{subequations}\label{theta-A-KP}
\begin{align}
& \theta_i=-(\zeta(k_i)-\zeta(l_i)) x -(\wp(k_i)-\wp(l_i)) y +\frac{1}{2}(\wp'(k_i)-\wp'(l_i))t +\theta_{i}^{(0)}, ~~
\theta_{i}^{(0)}\in \mathbb{C},\label{theta-KP}\\
& e^{a_{ij}}=A_{ij}=\frac{\sigma(k_i-k_j)\sigma(l_i-l_j)}{\sigma(k_i-l_j)\sigma(l_i-k_j)}.\label{A-ij-KP}
\end{align}
\end{subequations}
\end{theorem}

\begin{proof}
First, by virtue of the quasi-gauge property of bilinear equations %with respect to Lam\'e function
(see Proposition \ref{P-B1})
and making use of identity \eqref{id-abcd},
equation \eqref{bilinear-bf} can be derived from
\begin{equation*}
(D_x^4-4D_xD_t-12\wp(x-\hbox{$\sum_{i=1}^N l_i$})D_x^2+3 D_y^2)\wt \tau\cdot \wt\tau=0,
\end{equation*}	
where $\wt \tau=f \wt g$.

Next, to write $\wt \tau/\wt g$ into an explicit form, let us look at the generic term $\wt \tau_{J}^{}/\wt g$ in $f$.
It follows from Lemma \ref{L-5} that
\begin{equation}\label{f-tg-KP2}
\frac{\wt \tau_{J}^{}}{\wt g}
=\frac{\sigma(x+\sum_{i\in J} (k_i-l_i))}{\sigma(x)\prod_{i\in J}\sigma(k_i-l_i)}\cdot
\left(\prod_{i<j\in J}\frac{\sigma(k_i-k_j)}{\sigma(l_i- l_j)}\right)
\left(\prod_{i\in J}e^{\beta_i}\right)
\exp \left[-\sum_{i\in J} (\wt\gamma(k_i)-\wt\gamma(l_i))\right],
\end{equation}
where $\wt\gamma(k)=\gamma(k)|_{x\to x-\sum_{i=1}^{N}l_i}$ and
\begin{equation*}
e^{\beta_i}=\sigma(k_i-l_i)\frac{\sigma^N(l_i)}{\sigma^N(k_i)}
\prod_{j\in S\setminus J}\frac{\sigma(k_{i}- l_j)}{\sigma(l_{i}-l_j)}.
\end{equation*}
In particular, if $J$ contains a single element, e.g.  $J=\{i\}$, we have
\begin{equation*}
\frac{\wt \tau_{\{i\}}}{\wt g}
=\Phi_x(k_i-l_i)e^{\alpha_{i}} e^{-\wt\gamma(k_{i})+\wt\gamma(l_{i})},
\end{equation*}
where
\begin{equation*}
e^{\alpha_{i}}=\sigma(k_i-l_i)\frac{\sigma^N(l_{i})}{\sigma^N(k_{i})}
\prod_{j\in S\atop j\neq i}\frac{\sigma(k_{i}- l_j)}{\sigma(l_{i}-l_j)}.
\end{equation*}
Define $\theta^{(0)}_i=\alpha_i-\gamma^{(0)}(k_i)+\gamma^{(0)}(l_i)+\sum_{j=1}^Nl_j(\zeta(k_i)-\zeta(l_i))$
such that $e^{\alpha_{i}} e^{-\wt\gamma(k_{i})+\wt\gamma(l_{i})}=e^{\theta_i}$
where $\theta_i$ is defined as in \eqref{theta-KP}.
Then, the generic term \eqref{f-tg-KP2} in $f$ is written into
\begin{equation*}
\frac{\wt \tau_{J}^{}}{\wt g}
=\frac{\sigma(x+\sum_{i\in J} (k_i-l_i))}{\sigma(x)\prod_{i\in J}\sigma(k_i-l_i)}
\left(\prod_{i<j\in J}A_{ij}\right) \exp\left(\sum_{i\in J}\theta_i\right),
\end{equation*}
where we have made use of
\[\prod_{i \in J}e^{\beta_i-\alpha_i}=\prod_{i,j\in J \atop i\neq j}\frac{\sigma(l_{i}- l_j)}{\sigma(k_{i}-l_j)}
=\prod_{i<j\in J}\frac{\sigma^2(l_{i}- l_j)}{\sigma(k_{i}-l_j)\sigma(l_{i}-k_j)}\]
and $A_{ij}$ is defined as in \eqref{A-ij-KP}.
It then turns out that $f=\sum_{J\subset S}\wt \tau_{J}^{}/\wt g$ takes the explicit Hirota's form \eqref{f-Hirota-KP}.

\end{proof}

\subsection{Vertex operator}\label{sec-4-2}

We now present a vertex operator that can generate  $\tau$ functions for
elliptic soliton solutions of the KP hierarchy.
Introduce a vertex operator
\begin{equation}\label{eq:vertexkp}
X(k,l)=\Phi_{t_1}(k-l)e^{\xi_{[e]}(\mathbf{t},k)-\xi_{[e]}(\mathbf{t},l)}
e^{\xi(\wt \partial,k)-\xi(\wt \partial,l)},
\end{equation}
where $\xi$ and $\xi_{[e]}$ are defined in \eqref{xi-theta},
and
$\widetilde{\partial}=(\partial_{t_1},\frac{1}{2} \partial_{t_2},\cdots,\frac{1}{n} \partial_{t_n},\cdots)$.
Similar to the relation \eqref{Aij-pq} and Lemma \ref{L-2}, for $A_{ij}$ defined in \eqref{A-ij-KP},
it can be proved that
\begin{equation}\label{Aij-pq-KPP}
\ln A_{ij}=\xi_{[e]}(\varepsilon(k_j)-\varepsilon(l_j), k_i) - \xi_{[e]}(\varepsilon(k_j)-\varepsilon(l_j), l_i)
\end{equation}
and
\begin{subequations}
\begin{equation}\label{XX-Aij}
X(k_i,l_i)  X(k_j,l_j)= A_{i,j}\frac{\sigma(t_1+k_i-l_i+k_j-l_j)}{\sigma(t_1)\sigma(k_i-l_i)\sigma(k_j-l_j)}\,
:\!X(k_i,l_i)  X(k_j,l_j)\!:,
\end{equation}
where  $\varepsilon(q)=(q, \frac{q^2}{2}, \frac{q^3}{3},\cdots,\frac{q^{n}}{n},\cdots)$,
and by $:\!\!X\!\!:$ we denote the normalization of the exponential part of the vertex operator $X$
by moving all differential operators in $X$ to the right,
e.g., here we have
\begin{equation}
:\!X(k_i,l_i)  X(k_j,l_j)\!:\, =
e^{\xi_{[e]}(\mathbf{t},k_i)-\xi_{[e]}(\mathbf{t},l_i)}e^{\xi_{[e]}(\mathbf{t},k_j)-\xi_{[e]}(\mathbf{t},l_j)}
e^{\xi(\wt \partial,k_i)-\xi(\wt \partial,l_i)}e^{\xi(\wt \partial,k_j)-\xi(\wt \partial,l_j)}.
\end{equation}
\end{subequations}
A more general version of \eqref{XX-Aij} is
\begin{equation}\label{XXX-Aij}
\prod^{N}_{i=1}X(k_i,l_i) = \left(\prod_{1\leq i<j \leq N} A_{i,j}\right)
\frac{\sigma(t_1+\sum_{i=1}^N(k_i-l_i))}{\sigma(t_1)\prod_{i=1}^{N}\sigma(k_i-l_i)}\,:\! \prod^{N}_{i=1}X(k_i,l_i) \!:.
\end{equation}
It then follows that
\[X^2(k,l)=0,~~ e^{cX(k,l)}=1+cX(k,l),~~
e^{cX(k,l)}\circ 1=1+c\, \Phi_{t_1}(k-l) e^{\xi_{[e]}(\mathbf{t},k)-\xi_{[e]}(\mathbf{t},l)},\]
which leads us to the following result for elliptic $N$-soliton solution.

\begin{theorem}\label{T-8}
For the KP hierarchy, its $\tau$ function of elliptic $N$-soliton solution,
\begin{equation}\label{tau-KP}
\tau_N(\mathbf{t})=\sum_{J\subset S}\left(\prod_{i\in J}c_i\right)
\left(\mathop{\rm{\prod}}_{i<j \in J}A_{ij}\right)
\frac{\sigma(t_1+\sum_{i\in J}(k_i-l_i))}{\sigma(t_1)\prod_{i\in J}\sigma(k_i-l_i)}
e^{\sum_{i\in J}(\xi_{[e]}(\mathbf{t}, k_i)-\xi_{[e]}(\mathbf{t},l_i))},
\end{equation}
is generated by the vertex operator \eqref{eq:vertexkp} via
\begin{equation}
\tau_N(\mathbf{t})=e^{c_{N}X(k_N,l_N)}\cdots e^{c_{2}X(k_2,l_2)}e^{c_{1}X(k_1,l_1)}\circ 1,
\end{equation}
or via transformation
\begin{equation}
\tau_{N}(\mathbf{t})=e^{c_{N}X(k_{N},l_{N})}\circ \tau_{N-1}(\mathbf{t}),~ ~ \tau_{0}(\mathbf{t})=1.
\end{equation}
In addition, $\tau_N(\mathbf{t})$ is a doubly periodic function with respect to any $k_i$ and $l_j$
for $i,j=1,2, \cdots, N$.
\end{theorem}

The proof is similar to Theorem \ref{T-3} for the KdV equation and we skip it.
Note also that the single Lam\'e-type PWF of the KP hierarchy is
\begin{equation}\label{PWF-KP-e}
\rho= X(k,l) \circ 1= \Phi_{t_1}(k-l) e^{\xi_{[e]}(\mathbf{t},k)-\xi_{[e]}(\mathbf{t},l)}.
\end{equation}

\subsection{Bilinear identity}\label{sec-4-3}

Define two functions of $q$,
\begin{subequations}\label{h-q-KP}
\begin{align}
	 h (\mathbf{t},q)&= X(\mathbf{t},q)\tau (\mathbf{t}),
	\\
	 h^*(\mathbf{t},q)&= X^*(\mathbf{t},q)\tau (\mathbf{t}),
	\end{align}
\end{subequations}
where $\tau (\mathbf{t})=\tau_N^{} (\mathbf{t})$ is given by \eqref{tau-KP},
$X(\mathbf{t},q)$ and $X^*(\mathbf{t},q)$ are vertex operators
\begin{subequations}
\begin{align}
	 &X(\mathbf{t},q)=\frac{\sigma(t_1+q)}{\sigma(q)} e^{\xi_{[e]}(\mathbf{t},q)}
     e^{\xi(\wt{\mathbf{\partial}},q)},
	\\
	& X^*(\mathbf{t},q)=\frac{\sigma(t_1-q)}{\sigma(-q)} e^{-\xi_{[e]}(\mathbf{t},q)}
     e^{-\xi(\wt{\mathbf{\partial}},q)}.
	\end{align}
\end{subequations}
Similar to Lemma \ref{L-4} for the KdV case, we can write $h (\mathbf{t},q)$ and  $h^*(\mathbf{t},q)$
in their explicit forms,
\begin{align*}
h(\mathbf{t},q)=&\frac{\sigma(t_1+q)}{\sigma(q)}  e^{\xi_{[e]}(\mathbf{t},q)}  \\
	& \times	\sum_{J\subset S}\left(\mathop{\rm{\prod}}_{i<j\in J}A_{i,j}\right)
\frac{\sigma(t_1+\sum_{i\in J}(k_i-l_i)+q)}{\sigma(t_1+q)\prod_{i\in J}\sigma(k_i-l_i)}
		\left(\prod_{i\in J}c_i\frac{\sigma(l_i)\sigma(q-k_i)}{\sigma(k_i)\sigma(q-l_i)}
e^{\theta_{[e]}(\mathbf{t},k_i,l_i)} \right),\\
h^*(\mathbf{t},q)=&\frac{\sigma(t_1-q)}{\sigma(-q)}  e^{-\xi_{[e]}(\mathbf{t},q)}  \\
	& \times	\sum_{J\subset S}\left(\mathop{\rm{\prod}}_{i<j\in J}A_{i,j}\right)
\frac{\sigma(t_1+\sum_{i\in J}(k_i-l_i)-q)}{\sigma(t_1-q)\prod_{i\in J}\sigma(k_i-l_i)}
		\left(\prod_{i\in J}c_i\frac{\sigma(k_i)\sigma(q-l_i)}{\sigma(l_i)\sigma(q-k_i)}
e^{\theta_{[e]}(\mathbf{t},k_i,l_i)} \right),
\end{align*}
where $\theta_{[e]}(\mathbf{t},k_i,l_i)=\xi_{[e]}(\mathbf{t},k_i)-\xi_{[e]}(\mathbf{t},l_i)$.
Then it can be verified that both functions are doubly periodic with respect to $q$
with the same periods as $\wp(q)$.

Obviously, the double-periodic property yields a bilinear identity for the KP hierarchy.

\begin{theorem}\label{T-9}
For the functions $h(\mathbf{t},q)$ and $h ^*(\mathbf{t}',q)$ defined in \eqref{h-q-KP}, we have
\begin{equation}\label{bil-int-KP}
\oint_{\Omega} \frac{\mathrm{d}q}{2\pi i}\, h(\mathbf{t},q)\, h^*({\mathbf{ t}}',q) =0,
\end{equation}
which gives rise to
\begin{equation}\label{bil-res-KP}
\underset{q=0}{\mathrm{Res}}\left[ h(\mathbf{t},q)\, h^*({\mathbf{ t}}',q)\right]=0,
\end{equation}
where the  contour $\Omega$ takes the boundary, anticlockwise, of the open fundamental period parallelogram $\mathbb{D}$
(see Fig.\ref{Fig-1}) and all $\{k_i\}$ and $\{l_i\}$ are distinct and belong to $\mathbb{D}$.
\end{theorem}

\begin{proof}

The first identity \eqref{bil-int-KP} is obvious.

For the second one, first, note that the integrand $ h(\mathbf{t},q)\, h^*({\mathbf{ t}}',q)$
has only $2N$ isolated simple poles $\{ k_i\}_{i=1}^{N}$, $\{ l_i\}_{i=1}^{N}$,
and one isolated essential singularity $q=0$ in $\mathbb{D}$.
Then, for given $j_0\in S$, we are going to prove the following relation,
\begin{equation}\label{res-q-l}
\underset{q=k_{j_0}}{\mathrm{Res}}\left[ h(\mathbf{t},q)\, h^*({\mathbf{ t}}',q)\right]
=-\underset{q=l_{j_0}}{\mathrm{Res}}\left[ h(\mathbf{t},q)\, h^*({\mathbf{ t}}',q)\right].
\end{equation}
In fact, similar to the KdV case, for given $j_0\in S$, we have
\begin{align*}
&\underset{q=k_{j_0}}{\mathrm{Res}}\left[ h(\mathbf{t},q)\, h^*({\mathbf{ t}}',q)\right]\\
=&~ -\frac{c_{j_0}}{\sigma(k_{j_0})\sigma(l_{j_0})}
e^{\xi_{[e]}(\mathbf{t},k_{j_0})-\xi_{[e]}(\mathbf{t'},l_{j_0})}\\
	& \times	\!\!\sum_{J\subset S\backslash\{j_0\}}\! \left[%\left(\prod_{i\in J}c_i\right)
       \!\!\left(\prod_{i<j\in J}A_{ij}\!\right)\!\!
       \left(\prod_{i\in J}c_i\frac{\sigma(l_i)\sigma(k_{j_0}-k_i)}{\sigma(k_i)\sigma(k_{j_0}-l_i)} \right)
        \frac{\sigma(t_1+\hbox{$\sum_{i\in J}(k_i-l_i)$}+k_{j_0})}{\prod_{i\in J}\sigma(k_i-l_i)}
        \, e^{\sum_{i\in J}\theta_{[e]}(\mathbf{t},k_i,l_i)}\right]\\
     & \times \!\!  	\sum_{J\subset S\backslash\{j_0\}}\! \left[%\left(\prod_{i\in J}c_i\right)
       \!\!\left(\prod_{i<j\in J}A_{ij}\!\right)\!\!
       \left(\prod_{i\in J}c_i\frac{\sigma(k_i)\sigma(l_{j_0}-l_i)}{\sigma(l_i)\sigma(l_{j_0}-k_i)} \right)
        \frac{\!\sigma(t_1'+\hbox{$\sum_{i\in J}(k_i-l_i)$}-l_{j_0})}{\prod_{i\in J}\sigma(k_i-l_i)}
        \, e^{\sum_{i\in J}\theta_{[e]}(\mathbf{t}',k_i,l_i)}\right],
\end{align*}
and $\underset{q=l_{j_0}}{\mathrm{Res}}\left[ h(\mathbf{t},q)\, h^*({\mathbf{ t}}',q)\right]$
has the same form but with ``$+$'' sign instead. Thus, \eqref{res-q-l} holds and then \eqref{bil-res-KP} follows.

\end{proof}

In what follows, we derive bilinear hierarchy from the identity \eqref{bil-res-KP}.
We introduce $\tau'(\mathbf{t})=\sigma(t_1) \tau(\mathbf{t})$
and $\mathbf{t}=\mathbf{x}+\mathbf{y}$ and $\mathbf{t'}=\mathbf{x}-\mathbf{y}$,
where $\mathbf{x}=(x_1, x_2, x_3, \cdots)$, $\mathbf{y}=(y_1, y_2, y_3, \cdots)$.
Then, the bilinear identity \eqref{bil-res-KP} gives rise to
\begin{equation}\label{4-34}
\underset{q=0}{\mathrm{Res}}\left[ \frac{1}{\sigma^2(q)} e^{2\xi_{[e]}(\mathbf{y},q)}
\tau'(\mathbf{x}+\mathbf{y}+\varepsilon(q)) \tau'(\mathbf{x}-\mathbf{y}-\varepsilon(q))\right]=0,
\end{equation}
i.e.
\begin{equation}\label{4-35}
\underset{q=0}{\mathrm{Res}}\left[ \frac{1}{\sigma^2(q)} e^{2\xi_{[e]}(\mathbf{y},q)}
e^{(\mathbf{y}+\varepsilon(q))\cdot \mathbf{D}_{\mathbf{x}}}
\tau'(\mathbf{x}) \cdot \tau'(\mathbf{x}) \right]=0,
\end{equation}
which, by rearranging terms with respect to $\mathbf{y}^{\beta}$, is written as
\begin{equation}\label{bil-D-KP}
\sum_{|\beta|=0}^{\infty}\underset{q=0}{\mathrm{Res}}\left[
(\mathbf{B}+\mathbf{D}_{\mathbf{x}})^\beta
\left( \sum_{n=0}^{\infty}\sum_{j=0}^n p_j(\wt{\mathbf{D}}_{\mathbf{x}})\mu_{n-j} q^{n-2}\right)
\tau'(\mathbf{x}) \cdot \tau'(\mathbf{x})\right]\mathbf{y}^{\beta}=0.
\end{equation}
Here,
\begin{equation}\label{DBy-KP}
\begin{array}{ll}
\mathbf{D}_{\mathbf{x}}=(D_{x_1}, D_{x_2}, D_{x_3}, \cdots),~~
\wt{\mathbf{D}}_{\mathbf{x}}=(D_{x_1}, \frac{1}{2}D_{x_2}, \frac{1}{3}D_{x_3}, \cdots), \\
\mathbf{B}=2(-\zeta(q), \zeta'(q),-\frac{\zeta''(q)}{2!},  \cdots (-1)^n\frac{\zeta^{(n-1)}(q)}{(n-1)!}, \cdots),\\
\beta=(\beta_1, \beta_2, \beta_3, \cdots), ~~|\beta|=\sum_{j=1}^\infty \beta_j,~~
\mathbf{y}^\beta=y_1^{\beta_1}y_2^{\beta_2}\cdots,
\end{array}
\end{equation}
$\{p_j(\mathbf{x})\}$ are defined by \eqref{p-j} and $\{\mu_j\}$ by \eqref{sigma-expd}.
By a similar analysis as for the KdV case in Sec.\ref{sec-3-4},
we can formulate an algorithm for calculating residues at $q=0$, which gives rise to a bilinear KP hierarchy.

\begin{theorem}\label{T-10}
The bilinear KP hierarchy with elliptic solitons are given by
\begin{equation}\label{bil-D-KP-4}
\underset{q=0}{\mathrm{Res}}\left[
({\mathbf{B}}+\mathbf{D}_{\mathbf{x}})^\beta|_{\leq 1}
\left( \sum_{n=0}^{||\beta||-1}\sum_{j=0}^n p_j(\wt{\mathbf{D}}_{\mathbf{x}})\mu_{n-j} q^{n-2}\right)
\tau'(\mathbf{x}) \cdot \tau'(\mathbf{x})\right]=0,
\end{equation}
where $\beta$ stands for the set of nonnegative integers $(\beta_1, \beta_2, \cdots, \beta_{n}, 0, 0, \cdots)$,
and $(\mathbf{B}+\mathbf{D}_{\mathbf{x}})^\beta|_{\leq 1}$
means those  terms of $q^j$ with $j\leq 1$ in the Laurent series of $(\mathbf{B}+\mathbf{D}_{\mathbf{x}})^\beta$,
$||\beta||=\sum_{j=1}^nj\beta_{j}$
and $p_j(\mathbf{t})$ are polynomials defined by \eqref{p-j}.
\end{theorem}

Below are bilinear equations corresponding to
$\beta=(3,0,0,0,\cdots)$, $(4,0,0,0,\cdots)$, $(5,0,0,0,\cdots)$, $(3,1,0,0,\cdots)$ and $(2,0,1,0,\cdots)$,
respectively,
\begin{subequations}\label{bil-KP eqs}
\begin{align}
&(D_{x_1}^4+3 D_{x_2}^2-4 D_{x_1} D_{x_3}-g_2) \tau'\cdot \tau'=0,\\	
&(D_{x_1}^3D_{x_2}+2D_{x_2}D_{x_3}-3D_{x_1}D_{x_4}) \tau'\cdot \tau'=0,\\
&(D_{x_1}^6+45D_{x_1}^2D_{x_2}^2+20D_{x_1}^3D_{x_3}+40D_{x_3}^2+90D_{x_2}D_{x_4}
 -216D_{x_1}D_{x_5}+3g_2D_{x_1}^2-24g_3) \tau'\cdot \tau'=0,\\
&(D_{x_1}^6-45D_{x_1}^2D_{x_2}^2-20D_{x_1}^3D_{x_3}-80D_{x_3}^2
    +144D_{x_1}D_{x_5}+3g_2D_{x_1}^2-24g_3) \tau'\cdot \tau'=0,\\	
&(D_{x_1}^6-9D_{x_1}^2D_{x_2}^2+4D_{x_1}^3D_{x_3}-32D_{x_3}^2
    +36D_{x_2}D_{x_4}+3g_2D_{x_1}^2-24g_3) \tau'\cdot \tau'=0.
	\end{align}
\end{subequations}
When $g_2$, $g_3$ are 0,
these bilinear equations degenerate to the usual soliton case, cf.\cite{JM-RIMS-1983}.

\section{Degenerations and reductions}\label{sec-5}

In the following we investigate deformations of $\tau$ functions and bilinear equations
under the degenerations of periods and under the reductions of dispersion relations.

\subsection{Degenerations by periods}\label{sec-5-1}

When the invariants $g_2$ and $g_3$ of the elliptic curve \eqref{ell-cur} take
$g_2=\frac{4}{3}(\frac{\pi}{2 w_1})^4,~g_3=\frac{8}{27}(\frac{\pi}{2 w_1})^6$ and $g_2=g_3=0$,
the elliptic curve degenerates to be a cylinder and Riemann sphere, respectively.
These correspond to the degenerations from doubly periodic case to the singly period case and non-periodic case.
The Weierstrass functions will become trigonometric/hyperbolic functions and rational functions,
which we list in Proposition \ref{P-A2} in Appendix \ref{A-1}.
Obviously, such deformations hold in $\tau$ functions and bilinear equations.
In the following we present $\tau$ functions and bilinear equations of the trigonometric/hyperbolic case and rational case.
It is worth mentioning that we will give more concise formulae for the trigonometric/hyperbolic case.

\subsubsection{Trigonometric/hyperbolic case}\label{sec-5-1-1}

One can directly replace those Weierstarss functions in the bilinear form \eqref{4-35}
and $\tau$ function \eqref{tau-KP} using \eqref{weier-sp}.
As a result, for those explicit bilinear equations in \eqref{bil-KP eqs},
one needs to replace $g_2$ and $g_3$ by \eqref{const-g2g3},
and the $\tau$ function $\tau'$ is then given by
\[\tau'=e^{\frac{1}{6}(\alpha x_1)^2}\sin (\alpha x_1)\, \tau_N(\mathbf{x}),\]
where $\tau_N(\mathbf{x})$ is defined as in \eqref{tau-KP} but in which the Weierstarss functions are replaced accordingly
using \eqref{weier-sp}.

Such a $\tau_N(\mathbf{x})$ for the trigonometric/hyperbolic case can have a more concise form.
To achieve that, we introduce notation
\begin{equation}\label{xi-t}
\xi_{[t]}(\mathbf{x},k)=\alpha \sum^{\infty}_{n=1}(-1)^n x_n \frac{\partial^{n-1}_k \cot(\alpha k)}{(n-1)!},
\end{equation}
where by the index $[t]$ we indicate the trigonometric/hyperbolic case.
Then, similar to the formula \eqref{Aij-pq-KPP}, we can prove that
\begin{equation}\label{Aij-sin}
\frac{\sin(\alpha(k_i-k_j))}{\sin(\alpha(k_i-l_j))}
=e^{\xi_{[t]}(\varepsilon(k_j)-\varepsilon(l_j), k_i)}
\end{equation}
where  $\varepsilon(k)=(k, \frac{k^2}{2}, \frac{k^3}{3},\cdots)$ defined as before.
Next, we present a simple form of $\tau_N(\mathbf{x})$ and the related vertex operator.

\begin{theorem}\label{T-11}
The bilinear hierarchy \eqref{bil-D-KP-4} with degeneration \eqref{const-g2g3}
have a solution
\begin{equation}\label{tau-sp}
\tau'=e^{\frac{1}{6}(\alpha x_1)^2}\sin (\alpha x_1)\, \tau_N(\mathbf{x}),
\end{equation}
where
\begin{equation}\label{tau-KP-sp}
\tau_N(\mathbf{x})=\sum_{J\subset S}\left(\prod_{i\in J}c'_i\right)
\left(\mathop{\rm{\prod}}_{i<j \in J}A'_{ij}\right)
\frac{\sin(\alpha(x_1+\sum_{i\in J}(k_i-l_i)))}{\sin(\alpha x_1)\prod_{i\in J}\sin(\alpha(k_i-l_i))}\,
e^{\sum_{i\in J}(\xi_{[t]}(\mathbf{x}, k_i)-\xi_{[t]}(\mathbf{x},l_i))}.
\end{equation}
Here $c'_i\in \mathbb{C}$ and
\begin{equation}\label{Aij-KPP}
A'_{ij}=\frac{\sin(\alpha(k_i-k_j))\sin(\alpha(l_i-l_j))}{\sin(\alpha(k_i-l_j))\sin(\alpha(l_i-k_j))}.
\end{equation}
The related vertex operator is
\begin{equation}\label{eq:vertexkp-sp}
X(k,l)=\frac{\sin(\alpha(x_1+k-l))}{\sin(\alpha(x_1))\sin(\alpha(k-l))}e^{\xi_{[t]}(\mathbf{x},k)-\xi_{[t]}(\mathbf{x},l)}
e^{\xi(\wt \partial,k)-\xi(\wt \partial,l)}.
\end{equation}
The $\tau$ function \eqref{tau-KP-sp} is defined by the vertex operator via
\begin{equation}\label{tau-vo-sp}
\tau_N(\mathbf{x})=e^{c'_{N}X(k_N,l_N)}\cdots e^{c'_{2}X(k_2,l_2)}e^{c'_{1}X(k_1,l_1)}\circ 1,
\end{equation}
i.e.
\begin{equation}
\tau_{N}(\mathbf{x})=e^{c'_{N}X(k_{N},l_{N})}\circ \tau_{N-1}(\mathbf{x}),~ ~ \tau_{0}(\mathbf{x})=1.
\end{equation}
\end{theorem}

\begin{proof}
Let us look at the $\tau_N(\mathbf{x})$ defined in \eqref{tau-KP}
where $\mathbf{t}=\mathbf{x}$.
We will show that, with $\sigma, \zeta, \wp$ taking the form \eqref{weier-sp},
the $\tau_N(\mathbf{x})$  can be written as in \eqref{tau-KP-sp}.
First, for a single PWF, we have
\begin{align*}
 c_i \frac{\sigma(x_1+k_i-l_i)}{\sigma(x_1)\sigma(k_i-l_i)}
e^{\xi_{[e]}(\mathbf{x}, k_i)-\xi_{[e]}(\mathbf{x},l_i)}\left |_{\eqref{weier-sp}}\right.
= c'_i \frac{\sin(\alpha(x_1+ k_i-l_i))}{\sin(\alpha x_1)\sin(\alpha(k_i-l_i))}\,
e^{\xi_{[t]}(\mathbf{x}, k_i)-\xi_{[t]}(\mathbf{x},l_i)}
\end{align*}
where we take
\begin{equation}
c'_i=c_i e^{\frac{1}{6}\alpha^2(k_i-l_i)^2}.
\end{equation}
Secondly, for the general term in $\tau_N(\mathbf{x})$, we have
\begin{align*}
& \left(\prod_{i\in J}c_i\right)
\frac{\sigma(x_1+\sum_{i\in J}(k_i-l_i))}{\sigma(x_1)\prod_{i\in J}\sigma(k_i-l_i)}
e^{\sum_{i\in J}(\xi_{[e]}(\mathbf{x}, k_i)-\xi_{[e]}(\mathbf{x},l_i))}|_{\eqref{weier-sp}}\\
=& \left(\prod_{i\in J}c_i\right) e^{\frac{1}{6}\alpha^2(\sum_{i\in J}(k_i-l_i))^2}
\frac{\sin(\alpha(x_1+\sum_{i\in J}(k_i-l_i)))}{\sin(\alpha x_1)\prod_{i\in J}\sin(\alpha(k_i-l_i))}
e^{\sum_{i\in J}(\xi_{[t]}(\mathbf{x}, k_i)-\xi_{[t]}(\mathbf{x},l_i))}\\
=& \left(\prod_{i\in J}c'_i\right)e^{\frac{1}{3}\alpha^2 \sum_{i<j\in J}(k_i-l_i)(k_j-l_j)}
\frac{\sin(\alpha(x_1+\sum_{i\in J}(k_i-l_i)))}{\sin(\alpha x_1)\prod_{i\in J}\sin(\alpha(k_i-l_i))}
e^{\sum_{i\in J}(\xi_{[t]}(\mathbf{x}, k_i)-\xi_{[t]}(\mathbf{x},l_i))}.
\end{align*}
Thirdly, for the phase factor $A_{ij}$, we have
\[A_{ij}|_{\eqref{weier-sp}}
=\left.\frac{\sigma(k_i-k_j)\sigma(l_i-l_j)}{\sigma(k_i-l_j)\sigma(l_i-k_j)}\right|_{\eqref{weier-sp}}
=e^{-\frac{1}{3}\alpha^2 \sum_{i<j\in J}(k_i-l_i)(k_j-l_j)} A'_{ij}.\]
All these together lead us to the form \eqref{tau-KP-sp}
for the $\tau$ function \eqref{tau-KP} with \eqref{weier-sp}.

For the vertex operator \eqref{eq:vertexkp-sp}, using relation \eqref{Aij-sin}, one can find that
\begin{equation*}
X(k_i,l_i)  X(k_j,l_j)= A'_{i,j}\frac{\sin(\alpha(x_1+k_i-l_i+k_j-l_j))}
{\sin(\alpha(x_1))\sin(\alpha(k_i-l_i))\sin(\alpha(k_j-l_j))}\,:\!X(k_i,l_i)  X(k_j,l_j)\!:,
\end{equation*}
where
\begin{equation*}
:\!X(k_i,l_i)  X(k_j,l_j)\!:\, =
e^{\xi_{[t]}(\mathbf{x},k_i)-\xi_{[t]}(\mathbf{x},l_i)}e^{\xi_{[t]}(\mathbf{x},k_j)-\xi_{[t]}(\mathbf{x},l_j)}
e^{\xi(\wt \partial,k_i)-\xi(\wt \partial,l_i)}e^{\xi(\wt \partial,k_j)-\xi(\wt \partial,l_j)}.
\end{equation*}
Then, equation \eqref{tau-vo-sp} follows immediately.

\end{proof}

Compared with Theorem \ref{T-8}, it turns out that Theorem \ref{T-11} can be obtained from Theorem \ref{T-8}
 by formally  replacing $\sigma(x)$ and $\zeta(x)$ with $\sin(\alpha x)$ and $\alpha\cot(\alpha x)$.
This also agrees with the fully discrete case, cf. \cite{YN-JMP-2013}.
The trigonometric/hyperbolic PWF of the KP hierarchy is (cf. Eq.\eqref{PWF-KP-e})
\begin{equation}\label{PWF-KP-t}
\rho= X(k,l)\circ 1 =\frac{\sin(\alpha(x_1+k-l))}{\sin(\alpha(x_1))\sin(\alpha(k-l))}
e^{\xi_{[t]}(\mathbf{x},k)-\xi_{[t]}(\mathbf{x},l)}.
\end{equation}

\subsubsection{Rational case}\label{sec-5-1-2}

The $\tau$ function and vertex operator of rational case are obtained from Theorem \ref{T-8}
by direct substitution of \eqref{weier-np}.
Bilinear equations are those of doubly periodic case with degeneration $g_2=g_3=0$,
which are  the same as the bilinear equations for usual solitons.
We skip proof and only present main results in the following.

\begin{theorem}\label{T-12}
In the rational case the bilinear KP hierarchy are the same as the usual soliton case,
namely, the bilinear equations derived from \eqref{bil-D-KP} with $g_2=g_3=0$;
$\tau$ function is given by
\begin{equation}\label{tau-np}
\tau'=x_1\, \tau_N(\mathbf{x}),
\end{equation}
where
\begin{subequations}\label{tau-KPPP}
\begin{equation}\label{tau-KP-np}
\tau_N(\mathbf{x})=\sum_{J\subset S}\left(\prod_{i\in J}c_i\right)
\left(\mathop{\rm{\prod}}_{i<j \in J}A_{ij}\right)
\frac{x_1+\sum_{i\in J}(k_i-l_i)}{x_1\prod_{i\in J}(k_i-l_i)}\,
e^{\sum_{i\in J}(\xi_{[r]}(\mathbf{x}, k_i)-\xi_{[r]}(\mathbf{x},l_i))}.
\end{equation}
Here
\begin{align}\label{Aij-KP-np}
& A_{ij}=\frac{(k_i-k_j)(l_i-l_j)}{(k_i-l_j)(l_i-k_j)},\\
& \xi_{[r]}(\mathbf{x}, k)=-\sum^{\infty}_{n=1}\frac{1}{k^n}x_n,
\end{align}
\end{subequations}
and the subscript $[r]$ stands for the rational case.

The related vertex operator is
\begin{equation}\label{eq:vertexkp-np}
X(k,l)=\frac{x_1+k-l}{(k-l)x_1}e^{\xi_{[r]}(\mathbf{x},k)-\xi_{[r]}(\mathbf{x},l)}
e^{\xi(\wt \partial,k)-\xi(\wt \partial,l)},
\end{equation}
and the $\tau$ function \eqref{tau-KP-np} is generated  via
\begin{equation}\label{tau-vo-np}
\tau_N(\mathbf{x})=e^{c'_{N}X(k_N,l_N)}\cdots e^{c'_{2}X(k_2,l_2)}e^{c'_{1}X(k_1,l_1)}\circ 1.
\end{equation}
\end{theorem}

Note that the rational-type PWF of the KP hierarchy is (cf. Eq.\eqref{PWF-KP-e} and \eqref{PWF-KP-t})
\begin{equation}\label{PWF-KP-r}
\rho=X(k,l)\circ 1=\frac{x_1+k-l}{(k-l)x_1}e^{\xi_{[r]}(\mathbf{x},k)-\xi_{[r]}(\mathbf{x},l)}.
\end{equation}

\subsection{Reductions by dispersion relations}\label{sec-5-2}

\subsubsection{Elliptic case}\label{sec-5-2-1}

For the KP hierarchy, the vertex operator of its usual soliton solution is
\begin{equation}\label{eq:vertexkp-ss}
X(k,l)=e^{\xi(\mathbf{t},k)-\xi(\mathbf{t},l)}
e^{\xi(\wt \partial,k)-\xi(\wt \partial,l)},
\end{equation}
which is governed by $\xi(\mathbf{t},k)$.
Reduction by dispersion relation can be implemented through imposing constraints on $l$
such that $l^N=k^N$, i.e. $l=\omega k$ where $\omega$ is some $N$-th root of unity
and in practice we require $\omega^s\neq 1$ for $s=1,2,\cdots, N-1$.
The bilinear KP hierarchy together with its $\tau$ function will reduce to the lower dimension
for the Gel'fand-Dickey hierarchy,
including the KdV for $N=2$, the Boussinesq for $N=3$, etc.
For the case of elliptic solitons, however, the vertex operator \eqref{eq:vertexkp}
is governed by $\xi_{[e]}(\mathbf{t},k)$ and  $\xi(\wt\partial,k)$ together.
To implement reduction of elliptic solitons by dispersion relation,
one needs to make use of elliptic $N$-th roots of the unity, which is introduced in \cite{NSZ-2019}
(also see Definition \ref{D-1} in Appendix).

In the elliptic case, the $\tau$ function and bilinear equations of the KP hierarchy are reduced to
those of the KdV hierarchy by taking $l_j=-k_j$. This is because when $l=-k$
the the coordinate variables $t_{2n}$ in $\xi_{[e]}(\mathbf{t},k)-\xi_{[e]}(\mathbf{t},l)$
and $\partial_{t_{2n}}$ in  $\xi(\wt\partial,k)-\xi(\wt\partial,l)$
in the vertex operator \eqref{eq:vertexkp} vanish.

However, recalling the Remark \ref{R-A1} we give at the end of Appendix \ref{A-1},
except $\omega_0(\delta)\equiv \delta$, the other two elliptic cube roots of the unity are not the  elliptic 6-th roots of the unity.
This means, in principle, when $N\geq 3$ we cannot get elliptic $N$-soliton solution for the  Gel'fand-Dickey hierarchy
from those of the KP hierarchy by using elliptic $N$-th roots of the unity.

In the following we only present the $\tau$ function and bilinear equation for the Boussinesq equation (not the hierarchy),
which can be reduced from those of the KP equation using elliptic cube roots of the unity.
Let $\omega_0(\delta)\equiv \delta,~\omega_1(\delta)$ and $\omega_2(\delta)$ be three elliptic cube roots of the unity,
then
\begin{equation}\label{f-Hirota-BSQ}
f=\sum_{\mu=0,1} \frac{\sigma(x+\sum_{i=1}^N \mu_i
(k_i-\omega_1(k_i)))}{\sigma(x)\prod^N_{i=1}\sigma^{\mu_i}(k_i-\omega_1(k_i))}
\mathrm{exp}\left(\sum^{N}_{j=1} \mu_j \widehat\theta_j+\sum^N_{1\leq i<j}\mu_i\mu_j a_{ij}\right)
\end{equation}
is a solution of the bilinear Boussinesq equation
\begin{equation}
(D_x^4- 12\wp(x)D_x^2+3D_y^2)f \cdot f=0,
\end{equation}
where the summation of $\mu$ means to take all possible $\mu_i=\{0,1\}$  for $ i=1,2,\cdots, N$,
\begin{subequations}\label{theta-A-BSQ}
\begin{align}
& \widehat\theta_i=-(\zeta(k_i)-\zeta(\omega_1(k_i))) x -(\wp(k_i)-\wp(\omega_1(k_i))) y +\widehat\theta_{i}^{(0)}, ~~
\widehat\theta_{i}^{(0)}\in \mathbb{C}, \\
& e^{a_{ij}}=A_{ij}=\frac{\sigma(k_i-k_j)\sigma(\omega_1(k_i)-\omega_1(k_j))}
{\sigma(k_i-\omega_1(k_j))\sigma(\omega_1(k_i)-k_j)}.\label{A-ij-BSQ}
\end{align}
\end{subequations}
Note that it is easy to write out a vertex operator for the $\tau$ function \eqref{f-Hirota-BSQ}.
We skip it.

\subsubsection{Trigonometric/hyperbolic case}\label{sec-5-2-2}

Similar to the elliptic case, to consider reduction, we need to introduce trigonometric/hyperbolic $N$-th roots of the unity.
This can be done by considering period degeneration in Definition \ref{D-1}.
After suitable scaling of independent variables, we have the following.

\begin{definition}\label{D-2}
 There exist distinct $\{\omega_j(\delta)\}_{j=0}^{N-1}$,  up to the  periods $k\pi$ ,
such that the following equation holds,
\begin{equation}\label{root-tr}
\prod^{N-1}_{j=0}\Psi_{\kappa}(\omega_j(\delta))
=\frac{1}{(N-1)!}(\partial^{N-2}_{\kappa}\csc^{2}(-\kappa)-\partial^{N-2}_{\kappa}\csc^{2}(\delta))=0,
\end{equation}
where
\begin{equation}\label{Psi}
\Psi_{a}(b)=\frac{\sin(a+b)}{\sin(a) \sin(b)},
\end{equation}
$\omega_0(\delta)=\delta$ and all $\{\omega_j(\delta)\}$ are independent of $\kappa$.
$\{\omega_j(\delta)\}_{j=0}^{N-1}$ are called trigonometric/hyperbolic $N$-th roots of the unity.
\end{definition}

These roots also satisfy
\begin{equation}
\sum^{N-1}_{j=0}\omega_j(\delta)=0
\end{equation}
and
\begin{equation}
\sum^{N-1}_{j=0}\cot^{(l)}(\omega_j(\delta))=0,~~ (l=0,1,\cdots,N-2).
\end{equation}

When $N=2$, i.e. reduction to the KdV, we take $l_j=-k_j$ in the KP $\tau$ function \eqref{tau-KP-sp},
and we have the trigonometric/hyperbolic $\tau$ function of the KdV hierarchy:
\begin{equation}\label{tau-KdV-sp}
\tau_N(\mathbf{x})=\sum_{J\subset S}\left(\prod_{i\in J}c'_i\right)
\left(\mathop{\rm{\prod}}_{i<j \in J}A'_{ij}\right)
\frac{\sin(\alpha(x_1+2\sum_{i\in J}k_i))}{\sin(\alpha x_1)\prod_{i\in J}\sin(2\alpha k_i)}\,
e^{2\sum_{i\in J} \xi_{[t]}(\mathbf{x}, k_i)},
\end{equation}
where
\begin{equation}\label{Aij-KdV-sp}
A'_{ij}=\frac{\sin^2(\alpha(k_i-k_j))}{\sin^2(\alpha(k_i+k_j))},
\end{equation}
and by $\mathbf{x}$ we denote $(x_1,0,x_3,0,x_5,\cdots)$
for the sake of using the results of the KP hierarchy in Sec.\ref{sec-5-1-1}.
The above $\tau$ function is generated by vertex operator
\begin{equation}\label{eq:vertexkdv-sp}
X(k)=\frac{\sin(\alpha(x_1+2k))}{\sin(\alpha(x_1))\sin(2\alpha k)}e^{2\xi_{[t]}(\mathbf{x},k)}
e^{2\xi(\wt \partial,k)},
\end{equation}
where
$\widetilde{\partial}=(\partial_{x_1},0,\frac{1}{3} \partial_{x_3},0,\frac{1}{5} \partial_{x_5},\cdots)$.
Bilinear equations are those derived from \eqref{bil-D-KP}
by removing all $D_{x_{2n}}$ terms and imposing
$g_2=\frac{4}{3}\alpha^4$, $g_3=\frac{8}{27}\alpha^6$.
These equations have solution
\begin{equation}\label{tau-pr}
\tau'=e^{\frac{1}{6}(\alpha x_1)^2}\sin (\alpha x_1)\, \tau_N(\mathbf{x}),
\end{equation}
where $\tau_N(\mathbf{x})$ is given by \eqref{tau-KdV-sp}.

Same as the elliptic case, when $N \geq 3$ we cannot get $\tau$ function and bilinear equations of
the Gel'fand-Dickey hierarchy from those of the KP hierarchy by reduction
using triginametric/hyperbolic $N$-th roots of the unity.
For the Boussinesq equation (not hierarchy), it allows a $\tau$ function
\begin{equation}\label{f-Hirota-BSQ-tr}
f=\sum_{\mu=0,1}
\frac{\sin(\alpha(x_1+\mu_i(k_i-\omega_1(k_i))))}{\sin(\alpha x_1)\prod_{i=1}^N\sin^{\mu_i}(\alpha(k_i-w_1(k_i)))}
\mathrm{exp}\left(\sum^{N}_{j=1} \mu_j \widehat\theta_j+\sum^N_{1\leq i<j}\mu_i\mu_j a_{ij}\right),
\end{equation}
where the summation of $\mu$ means to take all possible $\mu_i=\{0,1\}$  for $ i=1,2,\cdots, N$,
\begin{subequations}\label{theta-A-BSQ-tr}
\begin{align}
& \widehat\theta_i=-\alpha(\cot(\alpha k_i)-\cot(\alpha \omega_1(k_i))) x
 -\alpha^2 (\csc^2(\alpha k_i)-\csc^2(\alpha\omega_1(k_i))) y +\widehat\theta_{i}^{(0)}, ~~
\widehat\theta_{i}^{(0)}\in \mathbb{C}, \\
& e^{a_{ij}}=A_{ij}=\frac{\sin(\alpha(k_i-k_j))\sin(\alpha(\omega_1(k_i)-\omega_1(k_j)))}
{\sin(\alpha(k_i-\omega_1(k_j)))\sin(\alpha(\omega_1(k_i)-k_j))},\label{A-ij-BSQ-tr}
\end{align}
\end{subequations}
$\alpha \omega_1(k)$ is one of trigonometric/hyperbolic cube root of the unity by Definition \ref{D-2},
i.e.
\[\partial_{\kappa}\csc^2(\kappa)|_{\kappa=\alpha k}=\partial_{\kappa}\csc^2(\kappa)|_{\kappa=\alpha \omega_1(k)}.\]
Such a $\tau$ function is a solution to the bilinear Boussinesq equation
\begin{equation}\label{bilinear-bf-tr}
(D_x^4 + 4 \alpha^2 D_x^2 -12\alpha^2 \csc^2(\alpha x)D_x^2+3D_y^2)f \cdot f=0.
\end{equation}
Note that it is easy to write out a vertex operator for the $\tau$ function \eqref{theta-A-BSQ-tr}. We skip it.

\subsubsection{Rational case}\label{sec-5-2-3}

Reduction of this case is as same as the usual soliton case.
For example, reductions $l_j=-k_j$ and $l_j=\omega k_j$ where $\omega^3=1, \omega\neq 1$
reduce the results in Sec.\ref{sec-5-1-2} of the KP hierarchy to the
KdV hierarchy and the Boussinesq hierarchy, respectively.
Note that for the KdV equation its solution of this case has been obtained via
the Marchenko integral equation in \cite{AC-PL-1979}
and a direct linearisation approach in \cite{FA-AIP-1982},
and now it is clear how these solutions originate from the elliptic soliton solutions.

\section{Conclusions and discussions}\label{sec-6}

We have established a bilinear framework for the elliptic soliton solutions that are composed by the Lam\'e type PWFs.
Employing the KdV equation and KP equation as examples, we presented
$\tau$ functions for these elliptic $N$-soliton solutions in Hirota's form,
and the corresponding vertex operators and bilinear identities.
An algorithm has been developed to calculate residues and obtain bilinear equations.
Such a framework allows degenerations to the trigonometric/hyperbolic and rational cases
when the invariants $g_2$ and $g_3$ are specified for one period and non-period.
Reductions by dispersion relations can be implemented using elliptic $N$-th roots of the unity,
but except the KdV hierarchy, the reductions of elliptic and  trigonometric/hyperbolic soliton solutions
are not applicable to the Boussinesq hierarchy and other higher order Gel'fand-Dickey hierarchies.

We would like to address some related  topics for  further consideration.
First, are there any  algebras to characterize this type of vertex operators?
In other words, are these vertex operators the representations of some algebras?
Date, Kashiwara and Miwa \cite{DKM-PJA-1981} found that
the vertex operator related to affine Lie algebra $A_1^{(1)}$ \cite{LepW-1978} can be used to
define a symmetry group of the KdV $\tau$ function.
This then built up a beautiful connection between integrable systems and affine Lie algebras via vertex operators
\cite{DKM-PJA-1981,DJKM-RIMS-1982a,JM-RIMS-1983,MJD-book-1999}.
However, so far we did not find any similar algebraic structures behind our vertex operators
(excluding the rational case).
The vertex operators \eqref{eq:vertexkdv} and \eqref{eq:vertexkp} can be considered
as elliptic deformations of the usual vertex operators of the KdV equation and KP equation.
Without algebraic structure, one can still investigate such deformations on vertex operators of
other integrable systems (e.g. \cite{DJKM-RIMS-1982a,JM-RIMS-1983}),
and in particular, of discrete integrable systems
(e.g. \cite{DJM-JPSJ-1982a,DJM-JPSJ-1982b,DJM-JPSJ-1983a,DJM-JPSJ-1983b,DJM-JPSJ-1983c}).
In addition, note that $u=-2\wp(x)$ is an initial solution in our scheme,
and meanwhile it is the 1-gap and 1-genus solution 
in light of the finite-gap integration approach \cite{Dub-FAA-1975,DN-JETP-1974}.
It would be interesting to make clear the eigenvalue distribution of the corresponding spectral problem
where the potential is elliptic multi-solitons,
and recover these elliptic soliton solutions form some analytic approach, e.g. the inverse scattering transform.
Finally, there are vertex representations for quantum affine algebras \cite{FJ-PNAS-1988}.
It would be also interesting if such elliptic deformations could be extended to
quantum vertex operators.

\vskip 30pt

\subsection*{Acknowledgments}

DJ Zhang is grateful to Prof. Frank W Nijhoff, Dr. Cheng Zhang and Prof. Honglian Zhang  for their
discussions and comments on elliptic functions, finite gap solutions and quantum algebras.
This project is supported by the NSF of China (Nos. 11631007 and 11875040)
and Science and Technology Innovation Plan of Shanghai (No.20590742900).

\vskip 30pt

\appendix

\section{Weierstrass functions}\label{A-1}

We collect some notations and properties  of the Weierstrass functions that we may use in the paper.
One may refer to \cite{bookAkhiezer},\cite{bookHJN} and \cite{NSZ-2019}.

Three Weierstrass functions $\zeta(z)$, $\wp(z)$ and $\sigma(z)$ are connected via
\[\zeta(z)=\frac{\sigma'(z)}{\sigma(z)},~~ \wp(z)=-\zeta'(z).\]
Among them only $\wp(z)$ is a truly elliptic function by the definition of an elliptic function,
i.e. meromorphic and doubly periodic.
By $w_1$ and $w_2$ we denote two half periods of $\wp(z)$. $\zeta(z)$ and $\sigma(z)$ are quasi-periodic
with respect to $w_i$, in the sense that
\begin{subequations}\label{periodicity}
\begin{align}
&\zeta(z+2 w_i)=\zeta(z)+2 \zeta(w_i), \\
&\sigma(z+2 w_i)=-\sigma(z)e^{2\zeta(w_i)(z+w_i)},\quad i=1,2.
\end{align}
\end{subequations}
It is easy to check the following holds.

\begin{Proposition}\label{P-A1}
For a generalized Lam\'e function $\frac{\sigma(a+q)}{\sigma(b+q)}e^{c\zeta(q)}$ where $a,b,c$ are constants,
it is doubly periodic with respect to $q$ if $a-b+c=0$.
\end{Proposition}

Let $e_i=\wp(w_i)$ for $i=1,2,3$ where $w_3=-w_1-w_2$.
$(\wp(z), \wp'(z))$ is a point on the Weierstrass elliptic curve
\begin{equation}\label{ell-cur}
y^2=R(x)=4x^3-g_2x-g_3=4(x-e_1)(x-e_2)(x-e_3),
\end{equation}
i.e.
\begin{equation}\label{ell-cur-2}
(\wp'(z))^2=4\wp^3(z)-g_2\wp(z)-g_3,
\end{equation}
where
$g_2=-4(e_1e_2+e_2e_3+e_3e_1)$ and $g_3=4e_1e_2e_3$ are invariants of the curve.
Differentiating \eqref{ell-cur-2} yields
\begin{equation}\label{ell-cur-3}
2\wp''(z)=12\wp^2(z)-g_2
\end{equation}
and further
\begin{equation}\label{ell-cur-4}
\wp^{(3)}(z)=12\wp(z)\wp'(z).
\end{equation}
The latter is the stationary KdV equation,
in other words, $u=-2 \wp(x)$ is a stationary solution to the KdV equation \eqref{kdv}.

$\wp(z)$ is an even function, while $\zeta(z)$ and $\sigma(z)$ are odd. $\sigma(z)$ is an entire function.
As for expansions, they have
\begin{subequations}\label{pzs-expan}
\begin{align}
& \wp(z)=\frac{1}{z^2}+\frac{g_2}{20}z^2 + \frac{g_3}{28}z^4 + O(z^6),\\
& \zeta(z)=\frac{1}{z}-\frac{g_2}{60}z^3  - \frac{g_3}{140}z^5 + O(z^7), \label{zeta-expd}\\
& \sigma(z)=z-\frac{g_2}{240}z^5 -\frac{g_3}{840}z^7+O(z^9).
\end{align}
\end{subequations}

Some useful identities of the Weierstrass functions are given below.
\begin{equation}\label{eq:add-1}
 \wp(z)-\wp(u)=-\frac{\sigma(z+u)\sigma(z-u)}{\sigma^2(z)\sigma^2(u)},
\end{equation}
\begin{equation}\label{eq:add-2}
\eta_u(z)=\zeta(z+u)-\zeta(z)-\zeta(u)=\frac{1}{2}\frac{\wp'(z)-\wp'(u)}{\wp(z)-\wp(u)},
\end{equation}
\begin{equation}\label{eq:add-3}
\wp(z)+\wp(u)+\wp(z+u)=\eta_u^2(z) 
\end{equation}
and
\begin{equation}\label{eq:add-4}
\chi_{u,v}(z)=\zeta(u)+\zeta(v)+\zeta(z)-\zeta(u+v+z)
=\frac{\sigma(u+v)\sigma(u+z)\sigma(z+v)}{\sigma(u)\sigma(v)\sigma(z)\sigma(z+u+v)}.
\end{equation}
The famous Frobenius-Stickelberger determinant (also known as elliptic van der Monde determinant) is \cite{FS80}
\begin{equation}\label{eq:FS-1}
	\begin{split}
	&|\mathbf{1},~ \wp(\mathbf{k}),~ \wp'(\mathbf{k}),~ \wp''(\mathbf{k}),~ \cdots,~\wp^{(n-2)}(\mathbf{k}) |\\
	=&(-1)^{\frac{(n-1)(n-2)}{2}}\left(\prod^{n-1}_{s=1}s!\right) \frac{\sigma(k_1+\cdots+k_n)
        \prod_{i<j}\sigma(k_i-k_j)}{\sigma^n(k_1)\sigma^n(k_2),\cdots\sigma^n(k_n)},
    \end{split}
\end{equation}
where $f{(\mathbf{k})}$ denotes a column vector with entries $f(k_j)$,
i.e.  $f{(\mathbf{k})}=(f(k_1), f(k_2), \cdots, f(k_n))^T$.
One more formula is (see (C.7) in \cite{ND-2016})
\begin{equation}\label{eq:ND16}
\prod_{j=1}^n\Phi_x(k_j)=\frac{(-1)^{n-1}}{(n-1)!}\Phi_x(k_1+\cdots+k_n)
\frac{|\mathbf{1},~ \wp(\mathbf{k}),~ \wp'(\mathbf{k}),~ \cdots,~\wp^{(n-2)}(\mathbf{k}) |}
{|\mathbf{1}, ~\eta_x(\mathbf{k}),~ \wp(\mathbf{k}),~ \wp'(\mathbf{k}),~ \cdots,~\wp^{(n-3)}(\mathbf{k}) |},
\end{equation}
where $\Phi_a(b)=\frac{\sigma(a+b)}{\sigma(a)\sigma(b)}$.

Degenerations of the Weierstrass functions take place when the discriminant is zero, i.e.
\begin{equation}
\Delta=g_2^3-27g_3^2=0.
\end{equation}
The degenerations are described as the following \cite{bookAkhiezer}.
\begin{Proposition}\label{P-A2}
With parametrisation
\begin{equation}\label{const-g2g3}
g_2=\frac{4}{3}\alpha^4,~~g_3=\frac{8}{27}\alpha^6,~~\alpha=\frac{\pi }{2 w},
\end{equation}
the Weierstrass functions degenerate to the trigonometric/hyperbolic case,\footnote{We do not
discriminate between trigonometric and hyperbolic cases, as $\alpha$ (or the period $2w$)
can be either real or pure imaginary, corresponding to the two cases ($g_3$ being positive or negative)
to define the period through elliptic integrals \cite{bookAkhiezer}.
}
\begin{subequations}\label{weier-sp}
\begin{align}
&\sigma(q)  =  \frac{1}{\alpha}e^{\frac{1}{6}(\alpha q)^2}\sin (\alpha q),\\
& \zeta(q) =\frac{1}{3}\alpha^2 q + \alpha \cot (\alpha q),\\
& \wp(q)=-\frac{1}{3}\alpha^2  + \alpha^2 \csc^2 (\alpha q).
\end{align}
\end{subequations}
And when $g_2=g_3=0$,
the Weierstrass functions degenerate to the rational case,
\begin{equation}\label{weier-np}
\sigma(q)=q,~~
\zeta(q) =\frac{1}{q},~~
\wp(q)=\frac{1}{q^2}.
\end{equation}
\end{Proposition}

In what follows we present the definition of elliptic $N$-th roots of the unity
that was introduced in \cite{NSZ-2019}.

\begin{definition}\label{D-1}
\cite{NSZ-2019} There exist distinct $\{\omega_j(\delta)\}_{j=0}^{N-1}$,  up to the periodicity of the periodic lattice,
such that the following equation holds,
\begin{equation}\label{root}
\prod^{N-1}_{j=0}\Phi_{\kappa}(\omega_j(\delta))
=\frac{1}{(N-1)!}(\wp^{(N-2)}(-\kappa)-\wp^{(N-2)}(\delta))=0,
\end{equation}
where $\omega_0(\delta)=\delta$ and all $\{\omega_j(\delta)\}$ are independent of $\kappa$.
$\{\omega_j(\delta)\}_{j=0}^{N-1}$ are called elliptic $N$-th roots of the unity.
\end{definition}

These roots also satisfy \cite{NSZ-2019}
\begin{equation}
\sum^{N-1}_{j=0}\omega_j(\delta)=0
\end{equation}
and
\begin{equation}
\sum^{N-1}_{j=0}\zeta^{(l)}(\omega_j(\delta))=0,~~ (l=0,1,\cdots,N-2).
\end{equation}

\begin{Remark}\label{R-A1}
In usual case if $\omega$ is a $n$-th root of the unity, it is also a $(kn)$-th root of the unity where
$k\in \mathbb{N}$.
This is not true in the elliptic case.
Note that the elliptic square roots of the unity are also the elliptic $2k$-th roots of the unity
because $\wp^{(2n)}(x)$ is even.
However, for the elliptic cube root of the unity, $\omega_1(\delta)\neq \delta$,
it is not an elliptic 6th-root of the unity.
In other words,  $\wp'(\omega_1(\delta))=\wp'(\delta)$ holds
does not guarantee that the validity of $\wp^{(4)}(\omega_1(\delta))= \wp^{(4)}(\delta)$,
%but we do not have $\wp^{(4)}(\omega_1(\delta))= \wp^{(4)}(\delta)$
where $\omega_1(\delta)\neq \delta$ (mod the periodic lattice).
In fact, using the formulae \eqref{ell-cur-2}, \eqref{ell-cur-3} and \eqref{ell-cur-4},
we have
\[\wp^{(4)}(\omega_1(\delta))-\wp^{(4)}(\delta)
=30(\wp'(\omega_1(\delta))-\wp'(\delta))(\wp'(\omega_1(\delta))+\wp'(\delta))+
12 g_2(\wp(\omega_1(\delta))-\wp(\delta)).\]
In the case $\wp'(\omega_1(\delta))=\wp'(\delta)$, it reduces to
\[\wp^{(4)}(\omega_1(\delta))-\wp^{(4)}(\delta))=12 g_2(\wp(\omega_1(\delta))-\wp(\delta)),\]
which does not vanish for arbitrary $\delta$ unless  $g_2=0$ or $\omega_1(\delta)=\delta$.
\end{Remark}

\section{Elliptic 1- and 2-soliton solutions and bilinear formulae}\label{A-2}

The purpose of this section is not only to show details of deriving elliptic 1-soliton and 2-soliton solutions of the KdV equation,
but also to explore some calculating formulae of the Lam\'e-type PWFs (cf. $e^{kx+lt}$)
under Hirota's $D$ operator.

The Lam\'e-type PWF defined in \eqref{PWF-2ss}, i.e.
\begin{equation}\label{PWF-appC}
\rho_i(x,t)=\Phi_x (x+2k_i) e^{\xi_i},~~
\xi_i=e^{-2\zeta(k_i)x+\wp'(k_i)t+\xi^{(0)}_i},
\end{equation}
satisfies the following relations
\begin{subequations}\label{pxs}
\begin{align}
	&\rho_{i,x}=-\chi_{k_i,k_i}(x)\rho_{i},\label{px}\\
	&\rho_{i,xx}=2\eta_{k_i}(x)\rho_{i,x},\label{pxx}\\
    &\rho_{i,xxx}=(6\wp(x)+2\wp(x+k_i)+4\wp(k_i))\rho_{i,x},\label{pxxx}
\end{align}
\end{subequations}
where $\eta_x(y)$ and $\chi_{x,y}(z)$ are defined in \eqref{eta} and \eqref{chi}.
There are equivalent expressions for these derivatives.
For example, noticing that
\[\chi_{k,k}(x)=\eta_{-k}(x+k)-\eta_k(x+k)\]
and making use of \eqref{eq:add-2}, we have
\begin{equation}\label{px-2}
\rho_{i,x}=\frac{-\wp'(k_i)}{\wp(x+k_i)-\wp(k_i)}\rho_i.
\end{equation}
Using this formula to replace $\wp(k_i)\rho_{i,x}$ in \eqref{pxxx} yields
\begin{equation}\label{pxxx-2}
\rho_{i,xxx}=6(\wp(x+k_i)+\wp(x))\rho_{i,x} +4\wp'(k_i)\rho_{i},
\end{equation}
which gives another expression of $\rho_{i,xxx}$.
To calculate $\rho_{i,xxxx}$, differentiating  \eqref{pxxx-2} once with respect to $x$ yields
\[\rho_{i,xxxx}=6(\wp'(x+k_i)-\wp'(-x))\rho_{i,x} +6(\wp(x+k_i)+\wp(x))\rho_{i,xx} +4\wp'(k_i)\rho_{i,x},\]
which then, by making use of \eqref{eq:add-2} and \eqref{pxx}, gives rise to a simpler form for $\rho_{i,xxxx}$,
\begin{equation}\label{pxxxx}
\rho_{i,xxxx}=12\wp(x)\rho_{i,xx} +4\wp'(k_i)\rho_{i,x}.
\end{equation}

Hirota's procedure for deriving usual solitons relies on the property
\[D^n_xD^m_t e^{ax+bt}\cdot e^{ax+bt}=0,~~ a,b\in \mathbb{C},
\]
but this does not hold any longer for the Lam\'e-type PWF $\rho_i$.
For example, one can verify that
\begin{equation}\label{Dxt}
D_xD_t \rho_i\cdot \rho_i=0
\end{equation}
but
\begin{equation}\label{Dx2}
D^2_x \rho_i\cdot \rho_i=2(\wp(x)-\wp(x+2k_i)) \rho_i^2,
\end{equation}
which is not zero.
In addition, using the expressions \eqref{pxx}, \eqref{px-2}, \eqref{pxxx-2}, \eqref{pxxxx} and formula \eqref{eq:add-3},
we have
\begin{equation}\label{Dx4}
D^4_x \rho_i\cdot \rho_i=12\wp(x)D^2_x \rho_i\cdot \rho_i,
\end{equation}
which does not vanish either.
There could be a more general result. We have checked the following formula,
\begin{equation}\label{B.9}
D_x^{2n} ~\varrho \cdot \varrho=\frac{\wp^{(2n-1)}(x)}{\wp'(x)}D_x^2 ~\varrho\cdot\varrho,
~~~~ \varrho=\Phi_x (a) e^{bx+ct},~~ a,b,c\in \mathbb{C},
\end{equation}
up to  $n=10$ using \textit{Mathematica}.
The `coefficient' $\frac{\wp^{(2n-1)}(x)}{\wp'(x)}$ is a linear combination of
$\{\wp^{s}(x)\}$ with  $s=n-1, n-3, n-4, \cdots, 1, 0$.
In checking the above relation we made use of the following formula (see Eq.(1.188) in \cite{Hir-book-2004})
\[2\cosh(\delta \partial_x) \ln \varrho=\ln(\cosh (\delta D_x) \varrho\cdot\varrho) \]
and $(\ln \varrho)_{xx}=\wp(x)-\wp(x+a)$.
However, a proof for arbitrary $n$ is absent.
Note that Hirota's $D$ operator allows gauge property with respect to linear exponential function, i.e.
\[D^n_xD^m_t (e^{ax+bt}f)\cdot (e^{ax+bt}g) =e^{2(ax+bt)} D^n_xD^m_t f\cdot g,\]
but the formula \eqref{B.9} indicates that such a property no longer holds when the linear exponential function
is replaced by the Lam\'e function.
Instead of that, we have the following.

\begin{Proposition}\label{P-B1}
(Quasi-gauge property) For the generalized Lam\'e-type PWF $\varrho$ defined in \eqref{B.9}
and $C^{\infty}$ functions $f(x,t)$ and $g(x,t)$, we have
\begin{equation}\label{quasi-gauge}
D^n_xD^m_t (\varrho f)\cdot (\varrho g) = \varrho^2  D^n_xD^m_t f\cdot g
+\sum^{\left[\frac{n}{2}\right]}_{l=1}\left(\begin{array}{c}n\\2l\end{array}\right)(D^{2l}_x \varrho\cdot \varrho)
D_x^{n-2l}D^m_t f \cdot g.
\end{equation}
In light of \eqref{B.9} the term $D^{2l}_x \varrho\cdot \varrho$ can be replaced by
$\frac{\wp^{(2l-1)}(x)}{\wp'(x)}D_x^2 \varrho\cdot \rho$ or
$\frac{2(\wp(x)-\wp(x+a))\wp^{(2l-1)}(x)}{\wp'(x)}\varrho^2$.
\end{Proposition}

\begin{proof}
The proof is direct by using the identity \cite{Hir-book-2004}
\begin{equation}\label{id-abcd}
\exp(D_1)(f h) \cdot (gk)=(\exp(D_1)h\cdot k)(\exp(D_1)f\cdot g),
\end{equation}
where $f,g,h,k$ are functions of $(x,t)$ and $D_1=\varepsilon D_x+\delta D_t$ with constants  $\varepsilon$ and $\delta$.

\end{proof}

We now look for elliptic soliton solutions in Hirota's form.
For elliptic 1-soliton solution with the form $f_1=1+\rho_1(x,t)$,
thanks to \eqref{Dxt} and \eqref{Dx4}, one only needs to verify
\begin{equation}\label{p-xt-1ss}
( \partial^4_x -4\partial_{xt}-12\wp(x) \partial^2_x)\rho_1=0,
\end{equation}
which is nothing but \eqref{pxxxx} in light of
\begin{equation}\label{pt}
\rho_{i,t}=\wp'(k_i) \rho_i.
\end{equation}
Thus, the elliptic 1-soliton \eqref{1ss-f1} is obtained.

Then we look for 2-soliton solution of the form
\begin{equation}
f_2 =1+\rho_1(x,t)+\rho_2(x,t)+f^{(2)}(x,t),
\end{equation}
subject to
\begin{equation}\label{bil-2ss}
(D_x^4-4D_tD_x-12\wp(x) D_x^2)f_2 \cdot f_2=0,
\end{equation}
where
\begin{equation}
 f^{(2)}(x,t)=A_{12} e^{4\zeta(k_1)k_2}\wt \rho_1 \rho_2,
 ~~~\wt \rho_1(x,t)=\rho_1(x+2k_2,t),
\end{equation}
and $A_{12}$ is a parameter to be fixed later.
In light of relations \eqref{Dxt} and \eqref{Dx4}, equation \eqref{bil-2ss} is reduced to two equations,
\begin{equation}\label{con-1}
(D_x^4-4D_tD_x-12\wp(x)D_x^2)~\rho_{1} \cdot \rho_{2} =
-(\partial^4_x-4\partial_{xt} -12\wp(x)\partial^2_x) f^{(2)}
\end{equation}
and
\begin{equation}\label{con-2}
(D_x^4-4D_tD_x-12\wp(x)D_x^2)~\rho_{i} \cdot f^{(2)} =0,  ~~  i=1,2.
\end{equation}

Let us first work on \eqref{con-1}.
By virtue of the fact \eqref{p-xt-1ss} which holds for $\rho_{2}$ as well, we have
\begin{align}
&(D_x^4-4D_tD_x-12\wp(x)D_x^2)~\rho_{1} \cdot \rho_{2}\nonumber\\
 =&
4\rho_{1,x}\rho_{2,t}+4\rho_{1,t}\rho_{2,x}
-4\rho_{1,xxx}\rho_{2,x}-4\rho_{1,x}\rho_{2,xxx}+6\rho_{1,xx}\rho_{2,xx}
+24 \wp(x)\rho_{1,x}\rho_{2,x}.\label{R-2}
\end{align}
Making use of \eqref{pt}, \eqref{pxx} and \eqref{pxxx},
we can express $\rho_{i,t}, \rho_{i,xx}$ and $\rho_{i,xxx}$ in terms of $\rho_{i,x}$.
After that, using formula \eqref{eq:add-3}, we arrive at
\begin{align}
(D_x^4-4D_tD_x-12\wp(x)D_x^2)~\rho_{1} \cdot \rho_{2}
 =&-12 (\eta_{k_1}(x)-\eta_{k_2}(x))^2\rho_{1,x}\rho_{2,x}\nonumber\\
=& -12 \chi_{-k_1,k_2}^2(x+k_1)\rho_{1,x}\rho_{2,x}, \label{R-3}
\end{align}
where use has been made of $\chi_{-k_1,k_2}(x+k_1)=\eta_{k_1}(x)-\eta_{k_2}(x)$.

For the right hand side of \eqref{con-1}, by virtue of \eqref{p-xt-1ss}, we have
\begin{align*}
& -(\partial^4_x -4\partial_{xt}-12\wp(x)\partial^2_x) f^{(2)}  \\
=&- A_{12}e^{4\zeta(k_1)k_2}
     \Bigl[-4\wt\rho_{1,x} \rho_{2,t}-4\wt\rho_{1,t} \rho_{2,x}
             +4 \wt\rho_{1,xxx} \rho_{2,x}+4\wt\rho_{1,x} \rho_{2,xxx}+6\wt\rho_{1,xx} \rho_{2,xx}\\
   & ~~~~~~~~~~~~~~~~ ~~~~~~~~~~
     +12(\wp(x+2k_2)-\wp(x))\wt\rho_{1,xx} \rho_{2}  -24 \wp(x) \wt\rho_{1,x} \rho_{2,x}\Bigr],
\end{align*}
in which
\[12(\wp(x+2k_2)-\wp(x))\wt\rho_{1,xx} \rho_{2}
=24 (\eta_{k_2}(x+k_2)-\eta_{k_2}(x))\eta_{k_1}(x+2k_2)   \wt\rho_{1,x} \rho_{2,x},
\]
where use has been made of \eqref{eq:add-3}, \eqref{px} and \eqref{pxx}.
Then, similar to the treatment for \eqref{R-2}, we have
\begin{align}
 -(\partial^4_x -4\partial_{xt} -12\wp(x)\partial^2_x) f^{(2)}
=& -12 A_{12}\,e^{4\zeta(k_1)k_2} \chi_{k_1,k_2}^2(x+k_2)\wt\rho_{1,x}\rho_{2,x} \nonumber\\
=& -12 A_{12}\, \chi_{k_1,k_2}^2(x+k_2)
       \frac{\Phi_{k_1}^2(x+2k_2)}{\Phi_{k_1}^2(x)}\rho_{1,x}\rho_{2,x},\label{R-4}
\end{align}
where we have used
\[\wt\rho_{1,x}=e^{-4\zeta(k_1)k_2} \frac{\Phi_{k_1}^2(x+2k_2)}{\Phi_{k_1}^2(x)}\rho_{1,x}.\]
Then, combining \eqref{R-3} and \eqref{R-4} together and expressing $\chi_{a,b}(c)$ in terms of $\sigma$
function using \eqref{eq:add-4}, we finally find
\[A_{12}=\frac{\sigma^2(k_1-k_2)}{\sigma^2(k_1+k_2)},\]
with which \eqref{con-1} holds.

Equation \eqref{con-2} can be verified straightforwardly.
The idea is as same as for verifying \eqref{con-1}, i.e.
using \eqref{p-xt-1ss} to eliminate those 4-th order derivatives of $\rho_{i}$ and $\wt\rho_{1}$,
and using \eqref{pt}, \eqref{pxx} and \eqref{pxxx} to express the equation
in terms of $\rho_{i,x}$ and $\wt\rho_{1,x}$.
After long and tedious calculation, we can verify \eqref{con-2} for $i=1,2$.
Thus, the elliptic 2-soliton solution \eqref{2ss} in Hirota's form is obtained.

In the above calculation, we expressed the bilinear equations in terms of $\rho_{i,x}$
and implemented verification by evaluating coefficients of $\rho_{1,x} \rho_{2,x}$, etc.
There is an alternative way to calculate bilinear derivatives of $\rho_{i}$ using the Bell polynomials.
Let us define
\begin{equation}\label{B.22}
\varrho_i=\Phi_x(a_i) e^{b_i x+c_i t}, ~~ a_i, b_i, c_i \in \mathbb{C}.
\end{equation}
Then we have
\begin{equation}
\varrho_{i,x}=\alpha_i(x) \varrho_i,
\end{equation}
where
\begin{equation}
\alpha_i(x) =\zeta(x+a_i)-\zeta(x)+b_i.
\end{equation}
Introduce functions
\begin{equation}
G_m(x)=\partial_x^{m-1}\alpha_1(x)  +(-1)^{m}\partial_x^{m-1}\alpha_2(x).
\end{equation}
Then, it can be proved that (see Eq.(3.4) in \cite{GilLNW-PRSLA-1996} and Eq.(10) in \cite{LamLSW-JPA-1994})
\begin{equation}
D^n_x D^m_t \varrho_{1}\cdot \varrho_{2}=(c_1-c_2)^m Y_n(G_1, G_2, \cdots, G_n) \varrho_{1} \varrho_{2},
\end{equation}
in which $Y_n$ is the Bell polynomials defined via (see Eq.(7.2) in \cite{Bell-AM-1934})
\[Y_n(y_1, y_2, \cdots, y_n)=e^{-y}\partial^n_x e^y,\]
where $y=y(x)$ is a $C^{\infty}$ function with respect to $x$ and $y_i$ stands for $\partial^i_x y(x)$.
$Y_n$ can be generated by
\[Y_n(y_1, y_2, \cdots, y_n)=
\sum \frac{n!}{\left(\prod^n_{s=1}c_s !\right) \left(\prod^n_{s=1} (s !)^{c_s}_{}\right)}
 \prod^n_{s=1}y_s^{c_s},
\]
where the sum is to be taken over all partitions of $n=\sum_{s=1}^n s c_s$.
The first few $\{Y_n\}$ are
\begin{align*}
& Y_0=1,~~ Y_1=y_1,~~ Y_2=y_2+y_1^2,\\
& Y_3=y_3+3y_1y_2 +y_1^3,\\
& Y_4=y_4+4 y_1y_3+3y_2^2 +6y_1^2y_2+y_1^4.
\end{align*}
The pioneer work that associates bilinearisation of soliton equations with the Bell polynomials is due to
\cite{GilLNW-PRSLA-1996,LamLSW-JPA-1994}.

For the Lam\'e-type function \eqref{B.22},
$Y_n(G_1, G_2, \cdots, G_n)$ is composed by functions such as $\zeta(x+a)$ and their derivatives
with respect to $x$, which might be finally converted to the expressions in terms of $\sigma$ function
by using the formulae given in Appendix \ref{A-1}.

\section{Proof of Theorem \ref{T-1} and Theorem \ref{T-6}} \label{A-3}

Before presenting the proof,
we recall two determinantal identities which are often used when verifying bilinear equations with Wronskian solutions.

\begin{Proposition}\label{P-C1}
\cite{FreN-1983} The  relation
	\begin{equation}
	|M,\mathbf{a},\mathbf{b}||M,\mathbf{c},\mathbf{d}|
-|M,\mathbf{a},\mathbf{c}| |M,\mathbf{b},\mathbf{d}|
+|M,\mathbf{a},\mathbf{d}| |M,\mathbf{b},\mathbf{c}|=0
	\end{equation}
holds, where $M$ is a $N\times(N-2)$ matrix, $\mathbf{a}, \mathbf{b}, \mathbf{c}$ and $\mathbf{d}$
are $N$-th order column vectors.
\end{Proposition}

\begin{Proposition}\label{P-C2}
\cite{ZhaZSZ-RMP-2014}	Suppose that $\Xi=(\Xi_{i,j})$ is a $N\times N$ matrix with column vector set $\{\Xi_j\}$,
$\Omega=(\Omega_{i,j})$ is a $N\times N$ operator matrix with column vector set $\{\Omega_j\}$
where entries are operators. Then we have
	\begin{equation}
	\sum_{j=1}^N |\Omega_j\ast \Xi|=\sum_{j=1}^N|(\Omega^T)_j \ast \Xi^T|,
	\end{equation}
where for any $N$-th order column vectors $A_j=(A_{1,j}, \cdots, A_{N,j})^T$ and
$B_j=(B_{1,j}, \cdots, B_{N,j})^T$ we define
	\begin{equation}
	A_j\circ B_j=(A_{1,j}B_{1,j}, A_{2,j}B_{2,j},\cdots, A_{N,j}B_{N,j})^T
	\end{equation}
	and
	\begin{equation}
	|A_j\ast\Xi|=|\Xi_1,\cdots,\Xi_{j-1},A_j\circ\Xi_j,\Xi_{j+1},\cdots,\Xi_{N}|.
	\end{equation}
\end{Proposition}

Now we start to prove Theorem \ref{T-1}.
For the $\tau$ function given in Wronskian form \eqref{tau-W},
where entries $\{\varphi_j\}$ obey relations \eqref{Wrons-cond},
by direct calculation, we have
\[\begin{array}{l}
		 \tau_x=|\widehat{N-2}, N|,\\
		 \tau_{xx}= |\widehat{N-3}, N-1 , N|+|\widehat{N-2}, N+1|,\\
		 \tau_{xxx}=|\widehat{N-4}, N-2 , N-1 , N|+2|\widehat{N-3}, N-1, N+1|+|\widehat{N-2}, N+2|,\\
		 \tau_{xxxx}=|\widehat{N-5}, N-3 , N-2 , N-1 , N|+3|\widehat{N-4}, N-2 , N-1 , N+1|\\
		  \qquad\qquad +2|\widehat{N-3},N , N+1|+3|\widehat{N-3}, N-1 , N+2|+|\widehat{N-2}, N+3|,\\
		\tau_t=|\widehat{N-4},N-2,N-1,N|-|\widehat{N-3},N-1,N+1|+|\widehat{N-2},N+2|\\
		 \qquad\qquad -\frac{3}{2}N^2\wp'(x)\tau -3\wp(x)\tau_x,\\
\tau_{tx}=|\widehat{N-2},N+3|-|\widehat{N-3},N,N+1|+|\widehat{N-5},N-3,N-2,N-1,N|\\
  \qquad\qquad  - \frac{3}{2} N^2\wp''(x)\tau-\frac{3}{2} (N^2+2)\wp'(x)\tau_x -3\wp(x)\tau_{xx}.
\end{array}
\]
Substituting them into the left hand side of  \eqref{bilinear-a} yields
\begin{equation}\label{C-5}
\begin{array}{l}
~~~4\tau_x\tau_t-4\tau_{xt}\tau+\tau_{xxxx}\tau-4\tau_{xxx}\tau_{x}+3\tau_{xx}^2-12\wp(x)(\tau_{xx}\tau-\tau_x^2)\\
=\tau \left(6N^2\wp''(x)\tau +12\wp'(x)\tau_x
 -3|\widehat{N-2},N+3|+6|\widehat{N-3},N,N+1| +3|\widehat{N-3},N-1,N+2|\right.\\
~~~	\left.	+3|\widehat{N-4},N-2,N-1,N+1|-3|\widehat{N-5},N-3,N-2,N-1,N|\right )\\
~~~		-12|\widehat{N-2},N| |\widehat{N-3},N-1,N+1|+3(|\widehat{N-3},N-1,N|+|\widehat{N-2},N+1|)^2.
\end{array}
\end{equation}
With the help of Proposition \ref{P-C2} where we take $\Omega_{j,s}=\wp(k_j)$,
from identity
$\left(\sum^{N}_{j=1}\wp(k_j) \tau \right)^2=\tau\sum^{N}_{j=1}\wp(k_j)\left(\sum^{N}_{j=1}\wp(k_j) \tau\right)$
we have
\begin{equation}\label{C-6}
\begin{array}{l}
0=\tau\left(-2N^2\wp''(x)\tau-4\wp'(x)\tau_x+|\widehat{N-2},N+3|+2|\widehat{N-3},N,N+1|
-|\widehat{N-3},N-1,N+2| \right.\\
~~~~\left. -|\widehat{N-4},N-2,N-1,N+1|+|\widehat{N-5},N-3,N-2,N-1,N|\right)\\
~~~~-(|\widehat{N-2},N+1|-|\widehat{N-3},N-1,N|)^2,
\end{array}
\end{equation}
using which equation \eqref{C-5} is reduced to
\begin{equation}\label{C-7}
\begin{array}{l}
		12(|\widehat{N-3},N,N+1| |\widehat{N-1}|-|\widehat{N-2},N| |\widehat{N-3},N-1,N+1|\\
+|\widehat{N-3},N-1,N| |\widehat{N-2},N+1|),
		\end{array}
\end{equation}
which vanishes in light of Proposition \ref{P-C1}. Thus Theorem \ref{T-1} is proved.

In a similar way we can prove Theorem \ref{T-6} for the KP equation.
In this case, the Wronskian entries $\{\varphi_j\}$ satisfy relation \eqref{eq:kplax}.
Derivatives of $\tau$ with respect $x$ and $t$ are the same as those for the KdV equation.
Besides them, we also have
\[	\begin{array}{l}
\tau_y=2 N\wp(x)\tau+|\widehat{N-3},N-1,N|-|\widehat{N-2},N+1|,\\
\tau_{yy}=4N^2\wp^2(x)\tau+4N\wp(x)(|\widehat{N-3},N-1,N|-|\widehat{N-2},N+1|)
                -4\wp'(x)\tau_x-2N^2\wp''(x)\tau\\
\qquad\qquad +|\widehat{N-5},N-3,N-2,N-1,N|
	 +2|\widehat{N-3},N,N+1|-|\widehat{N-3},N-1,N+2|\\
\qquad\qquad -|\widehat{N-4},N-2,N-1,N+1|+|\widehat{N-2},N+3|.
\end{array}
\]
For the KP equation, we do not have identity \eqref{C-6}. However,
$\tau_{yy}\tau-\tau_y^2$ contributes the same terms as the right hand side of \eqref{C-6}.
It then follows that
\begin{align*}
(D_x^4-4D_tD_x-12\wp(x)D_x^2+3D^2_y)\tau \cdot \tau
\end{align*}
is reduced to \eqref{C-7} as well, which is zero.
Thus, we complete the proof for Theorem \ref{T-6}.

\end{document}